\documentclass[11pt,fleqn,reqno]{article} 
\usepackage[margin=1.37in]{geometry}
\usepackage{amsmath,amsthm,amssymb,mathrsfs}
\usepackage{enumerate}
\usepackage[usenames,dvipsnames]{pstricks}
\usepackage{epsfig}
\usepackage{pst-grad} 
\usepackage{pst-plot}
\usepackage{hyperref} 
\usepackage{graphics,geometry,epsfig}

\theoremstyle{plain}
\newtheorem{theorem}{Theorem}
\newtheorem{proposition}[theorem]{Proposition}
\newtheorem{lemma}[theorem]{Lemma}
\newtheorem{corollary}[theorem]{Corollary}

\theoremstyle{definition}
\newtheorem{definition}[theorem]{Definition}

\newcommand{\R}{\mathbb{R}}
\newcommand{\C}{\mathbb{C}}
\newcommand{\Z}{\mathbb{Z}}
\newcommand{\bH}{\mathbb{H}}

\newcommand{\cE}{\mathcal{E}}

\newcommand{\cH}{\mathcal{H}}

\newcommand{\cL}{\mathcal{L}}         

\newcommand{\sH}{\mathscr{H}}

\newcommand{\nc}{\newcommand}

\nc{\G}{\Gamma}
\nc{\g}{\gamma}
\nc{\al}{\alpha}
\nc{\be}{\beta}
\nc{\del}{\delta}
\nc{\io}{\iota}
\nc{\ka}{\kappa}
\nc{\lam}{\lambda}
\nc{\Lam}{\Lambda}
\nc{\w}{\omega}
\nc{\om}{\omega}
\nc{\Om}{\Omega}
\nc{\Oms}{\Omega^*}
\nc{\s}{\sigma}
\nc{\Si}{\Sigma}
\nc{\ta}{\tau}
\nc{\h}{\theta}
\nc{\z}{\zeta}

\newcommand{\e}{\epsilon}

\newcommand{\lat}{\mathcal{L}} 
\newcommand{\Lat}{\mathcal{L}}
\newcommand{\LAT}{\mathcal{L}}
\newcommand{\LATo}{\mathcal{L}_\omega}

\newcommand{\LATt}{\mathcal{L^\tau}}

\nc{\oP}{\overline P}

\newcommand\F[1]{\mathbb{#1}}

\newcommand{\dbar}{\overline{\partial}}

\nc{\ran}{\rangle}
\nc{\lan}{\langle}

\newcommand{\im}{\operatorname{Im}}
\renewcommand{\Re}{\operatorname{Re}}
\renewcommand{\Im}{\operatorname{Im}}
\newcommand{\ra}{\rightarrow}
\newcommand{\Ran}{\operatorname{Ran}}

\newcommand{\Null}{\operatorname{Null}}
\newcommand{\nul}{\operatorname{Null}}

\newcommand{\diag}{\operatorname{diag}}

\newcommand{\ls}{\lesssim}

\newcommand{\one}{\mathbf{1}}

\nc{\bfone}{{\bf 1}}

\newcommand{\n}{\nabla}
\newcommand{\p}{\partial}

\newcommand{\Curl}{\operatorname{curl}}
\newcommand{\Div}{\operatorname{div}}
\newcommand{\curl}{\operatorname{curl}}
\newcommand{\divv}{\operatorname{div}}
\newcommand{\Divv}{\operatorname{Div}}
\newcommand{\Cov}[1]{\nabla_{\!\!#1}}

\newcommand{\DETAILS}[1]{}

\newcommand{\Lpsi}[2]{\mathscr{L}_{#2}^{2}} 
\newcommand{\LA}[2]{\vec{\mathscr{L}}_{}^{2}}
\newcommand{\HA}[2]{\vec{\mathscr{H}}_{}^{2}} 
\newcommand{\HAs}[2]{\vec{\mathscr{H}}_{}^{#2}} 

\numberwithin{theorem}{section}
\numberwithin{equation}{section}

\begin{document}

\title{On Abrikosov Lattice Solutions of the Ginzburg-Landau Equations} 
\date{January 11, 2017} 

\author{Li Chen\footnote{Dept of Mathematics, Univ of Toronto, Toronto, Canada; nehcili@gmail.com},\ Panayotis Smyrnelis\footnote{Centro de Modelamiento Matem\'{a}atico (UMI 2807 CNRS), Universidad de Chile, 
Santiago, Chile; 
psmyrnelis@dim.uchile.cl},\ Israel Michael Sigal\footnote{Dept of Mathematics, Univ of Toronto, Toronto, Canada; im.sigal@utoronto.ca}}

\maketitle

\begin{abstract}
We prove existence of Abrikosov vortex lattice solutions of the Ginzburg-Landau equations of superconductivity, with multiple 
 magnetic flux quanta per a fundamental cell. 
We also revisit the existence proof for the Abrikosov vortex lattices, 
streamlining some arguments and providing some essential details missing in earlier proofs for a single  magnetic flux quantum per a fundamental cell.  
 \medskip


\end{abstract}



\section{Introduction}
\label{sec:intro}



\paragraph{1.1 The Ginzburg-Landau equations.} The Ginzburg-Landau model of superconductivity describes a superconductor contained in
$\Omega \subset \R^n$, $n = 2$ or $3$, in terms of a complex order parameter $\Psi : \Omega \to \C$, and a
magnetic potential $A : \Omega \to \R^n$\footnote{The Ginzburg-Landau theory is reviewed in every book on superconductivity and most of the books on solid state or condensed matter physics.
For reviews of rigorous results see the papers \cite{CHO, DGP, GST, Sig} and the books \cite{SS, FH, JT, Rub}}. 
The Ginzburg-Landau theory specifies that  the difference between the superconducting and normal free energies\footnote{In the problem we consider here it is appropriate to deal with  Helmholtz free energy at a fixed average magnetic field $b:=\frac{1}{|\Omega|}\int_\Omega \curl{A},$ where $|\Omega|$ is the area or volume of $\Omega$.} in a state $(\Psi, A)$ is
\begin{equation}
\label{eq:GL-energy}
    E_\Omega(\Psi, A) := \int_\Omega |\Cov{A}\Psi|^2 + |\Curl A|^2 + \frac{\kappa^2}{2} (1 - |\Psi|^2)^2,
\end{equation}
where $\Cov{A}$ is the covariant derivative defined as $\nabla - iA$ and $\kappa$ is a positive constant that depends on the material properties of the superconductor and is called the Ginzburg-Landau parameter. In the case $n = 2$, $\Curl A := \frac{\partial A_2}{\partial x_1} - \frac{\partial A_1}{\partial x_2}$ is a scalar-valued function.
It follows from the Sobolev inequalities that for bounded open sets $\Omega$, the energy $E_\Omega$ is well-defined and $C^\infty$ as a
functional on the Sobolev space $H^1$.

The critical points of this functional must satisfy the well-known Ginzburg-Landau equations
inside $\Omega$:
\begin{subequations} \label{GLE}
    \begin{align}
    \label{GLEpsi}
   &   \Delta_A \Psi = \kappa^2(|\Psi|^2-1)\Psi,\\
    \label{GLEA}
  &      \Curl^*\Curl A = \Im(\bar{\Psi}\Cov{A}\Psi).
    \end{align}
\end{subequations}
Here $\Delta_A=- \Cov{A}^*\Cov{A},\ \Cov{A}^*$ and $\Curl^*$ are the adjoints of $\Cov{A}$ and $\Curl$. Explicitly, $\Cov{A}^*F = -\Div F + iA\cdot F$, and
$\Curl^* F = \Curl F$ for $n = 3$ and $\Curl^* f = (\frac{\partial f}{\partial x_2}, -\frac{\partial f}{\partial x_1})$ for $n = 2$.

The key physical quantities for the Ginzburg-Landau  theory are
\begin{itemize}
    \item the density of superconducting pairs of electrons, $n_s := |\Psi|^2$;
    \item the magnetic field, $B := \Curl A$;
    \item and the current density, $J := \Im(\bar{\Psi}\Cov{A}\Psi)$.
\end{itemize}

%

All superconductors are divided into two classes with different properties: Type I superconductors, which have $\kappa < \kappa_c$ and exhibit first-order phase transitions from the
non-superconducting state to the superconducting state, and Type II superconductors, which have $\kappa > \kappa_c$ and exhibit
second-order phase transitions and the formation of vortex lattices. Existence of these vortex lattice solutions is the subject of the present paper.


\paragraph{1.2 Abrikosov lattices.} In 1957, Abrikosov \cite{Abr} discovered solutions of \eqref{GLE} in $n=2$ whose physical characteristics
$n_s$, $B$, and $J$ are (non-constant and) periodic with respect to a two-dimensional lattice, while independent of the third dimension, and which have a single
flux per lattice cell\footnote{Such solutions correspond to cylindrical geometry.}.
We call such solutions the  \textit{($\lat-$)Abrikosov vortex lattices} or  the  \textit{($\lat-$)Abrikosov lattice solutions} or an abbreviation of thereof. 
(In physics literature they are variously called mixed states, Abrikosov mixed states, Abrikosov vortex
states.) Due to an error of calculation Abrikosov concluded that the lattice which gives the minimum average energy per
lattice cell\footnote{Since for lattice solutions the energy over $\R^2$ (the total energy) is infinite, one considers the average energy per lattice cell,
i.e. energy per lattice cell divided by the area of the  cell.} is the square lattice. Abrikosov's error was corrected by Kleiner, Roth, and Autler \cite{KRA}, who showed that it is in fact the triangular lattice which minimizes the energy.

Since their discovery, Abrikosov lattice solutions have been studied in numerous experimental and theoretical works.
Of more mathematical studies, we mention the articles of Eilenberger \cite{Eil}, Lasher \cite{Lash}, Chapman \cite{Ch} and Ovchinnikov \cite{Ov}.

The rigorous investigation of Abrikosov solutions began soon after their discovery. Odeh \cite{Odeh} sketched a proof of existence for 
various lattices using variational 
and bifurcation techniques.
Barany, Golubitsky, and Turski \cite{BGT} 
applied equivariant bifurcation theory and filled in a number of details, 
and Tak\'{a}\u{c} \cite{Takac} has adapted these results to study the zeros of the bifurcating solutions. Further details and new results, in both, variational and bifurcation, approaches, 
 were provided by \cite{Dutour2, Dutour}.  In particular, \cite{Dutour2} proved partial results on the relation between the bifurcation parameter and the average magnetic field $b$ (left open by previous works) and on the relation between  the Ginzburg-Landau energy and the Abrikosov function, and   \cite{Dutour}  (see also \cite{Dutour2}) found boundaries between superconducting, normal and mixed phases.

Among related results, a relation of the Ginzburg-Landau minimization problem, for a fixed, finite domain and external magnetic field,  
in the regime of $\kappa\ra \infty$, to the Abrikosov lattice variational problem was obtained in \cite{AS, Al2}.

The above investigation was completed and extended in \cite{TS, TS2}. 
To formulate  the results of these papers, we introduce some notation and definitions. 
   We define the following function on lattices $\Lat\subset \R^2$:
     \begin{equation}\label{kappac}
     \kappa_c(\Lat) := \sqrt{\frac{1}{2}\left(1-\frac{1}{\beta (\Lat)}\right)} ,
     \end{equation}
where $\beta(\Lat)$  is  the Abrikosov parameter, see e.g. \cite{TS, TS2}. 
 For a lattice $\Lat \subset \R^2$, we denote by $\Omega^\Lat$ and $|\Omega^\Lat|$ the basic lattice cell and its area, respectively. 
The following results were proven in \cite{TS, TS2}: 

\begin{theorem}\label{thm:main-resultTS}
For every lattice $\LAT$ satisfying    
\begin{equation}\label{LAT-cond}\big| 1 -b/\kappa^2 \big| \ll 1$ and $(\kappa-\kappa_c(\Lat))(\kappa^2-b)\ge 0$, where $b: = \frac{2\pi n}{|\Omega^\Lat|},\end{equation}
with $n=1$,  the following holds
	\begin{enumerate}[(I)]
	\item 
The equations \eqref{GLE} have an $\lat-$Abrikosov lattice solution 
in a neighbourhood of the branch of normal solutions.  
	\item The above solution 
is unique, up to symmetry, in a neighbourhood of the  normal branch.
	\item For $(\kappa-\kappa_c(\Lat))(\kappa^2-b)\ne 0$, the solution above is real analytic in $b$ in a neighbourhood of $\kappa^2$.
	\end{enumerate}	
\end{theorem}
Due to the flux quantization (see below), the quantity $b: = \frac{2\pi }{|\Omega^\Lat|}$, entering the theorem, is the average magnetic flux per lattice cell, $b := \frac{1}{|\Omega^\Lat|} \int_{\Omega^\Lat} \Curl A$.
We note that due to the reflection symmetry of the problem we can assume that $b \geq 0$.

\DETAILS{\begin{theorem}\label{thm:main-resultTS2}
Under the conditions of Theorem \ref{thm:main-resultTS} 
 and for $\kappa^2 > 1/2$, the lattice shape for which the average energy per lattice cell is minimized approaches the hexagonal lattice as $b \to \kappa^2$, in the sense that the shape parameter, $\tau_\cL $, of $\cL$ (see Subsection 3.3 below) approaches $ \tau_{triangular} = e^{i\pi/3}$ in $\C$.\end{theorem}}

All the rigorous results proven so far deal with Abrikosov lattices with one quantum of magnetic flux per lattice cell. Partial results for higher magnetic fluxes were proven in \cite{Ch, Al}. 


{\paragraph{1.3 Result.} 
  In this paper, we  prove existence of Abrikosov vortex lattice solutions of the Ginzburg-Landau equations, with multiple
 magnetic flux quanta per a fundamental cell, for certain lattices  and for certain flux quanta numbers.

We also revisit the existence proof for the Abrikosov vortex lattices, 
streamlining some arguments and providing some essential details missing in earlier proofs for a single  magnetic flux quantum per a fundamental cell.   
  
    As in the previous works, we consider only bulk superconductors filling all $\R^3$, with no variation along one direction, so that the problem is reduced to one on $\R^2$.

To formulate our results, we need some definitions. Motivated by the idea that most stable (i.e. most physical) solutions are also most symmetric, 
  we look for solutions which are most symmetric among vortex lattice solution for a given lattice and given the number of the flux quanta per fundamental cells.
Following \cite{DFN}, we denote
\begin{enumerate}
	\item $G(\cL)$ to be the group of symmetries of the lattice $\cL$,
	\item $T(\cL)$ to be the subgroup of $G(L)$ consisting of lattice translations, and
	\item $H(\cL) := 
G(\cL)\cap	O(2) \approx G(\cL)/T(\cL)$,  the maximal non-translation subgroup. 
\end{enumerate}

Note that the non-$SO(2)$ part of $H(\cL):=G(\cL)\cap O(2)$ comes from reflections and since all reflections in $G(\cL)\cap O(2)$ can be obtained as products of rotations and one fixed reflection, which we take to be $z \mapsto \bar{z}$, it suffices for us to consider the conjugation action. Since the conjugation is not holomorphic, we show in Section \ref{sec:no-sol} that there is no solutions having this symmetry. This  implies that the maximal point symmetry group of the GL equations is \[SH(\cL):=G(\cL)\cap SO(2) (=H(\cL)\cap SO(2)).\]

Hence, we look for solutions among functions whose physical properties are invariant under action of $SH(\cL)$. 

\begin{definition}[Maximal symmetry]
We say 
vortex lattice solution on $\R^2$ is \textit{maximally symmetric} iff all related physical quantities (i.e. $n_s := |\Psi|^2$, $B := \Curl A$, $J := \Im(\bar{\Psi}\Cov{A}\Psi)$) are invariant under the action of the group $SH(\cL)$, where $\lat$ is the underlying lattice of the solution.  
\end{definition}
Furthermore, we are interested in 
vortex lattice solutions with the following natural property 
\begin{definition}[$\cL-$ irreducibility] We say that a solution is $\cL-$irreducible iff there are no finer lattice for which it is a vortex lattice solution. 
 \end{definition}

Our main result is the following  
\begin{theorem} \label{thm:MultiFluxExist}
Assume either $n$ is one of $2,4,6,8,10$ and $\LAT$ is a hexagonal lattice or $n=3$ 
  and $\LAT$ is arbitrary. Then the GLEs have an  $\LAT-$irreducible, 
   maximally symmetric solution brach $(\lambda(s), \Psi(s), A(s)),\ s \in \F{R}$ small. This branch is, after rescaling \eqref{resc}, of the form \eqref{s-expansions}.  
  \end{theorem}
 
As was mentioned above, we revisit the existence proof of \cite{TS, TS2} 
streamlining some arguments and providing some essential details either missing or only briefly mentioned (\cite{TS, TS2}) in earlier proofs of the existence  of Abrikosov vortex lattices. 
  
 After introducing general properties of \eqref{GLE} in Sections \ref{sec:problem}-\ref{sec:rescaling}, we prove an abstract conditional result in Sections \ref{sec:operators}-\ref{sec:bifurc-dimK=1}, from which we derive  Theorem \ref{thm:main-resultTS} 
 in Section \ref{sec:bifurc-n=1} (giving a streamlined proof of this result) and Theorem \ref{thm:MultiFluxExist}, in Section \ref{sec:bifurc-point-sym}.



\bigskip

\noindent \textbf{Acknowledgements} It is a pleasure to thank Max Lein for useful discussions.  The first and third authors' research is supported in part by NSERC Grant No. NA7901. During the work on the paper, they enjoyed the support of the NCCR SwissMAP.
The second author (P. S.) was partially supported by Fondo Basal CMM-Chile and Fondecyt postdoctoral grant 3160055


\section{Properties of the  Ginzburg-Landau equations} \label{sec:problem}

\textbf{2.1 Symmetries.} The Ginzburg-Landau equations exhibit a number of symmetries, that is, transformations which map solutions to solutions: 

\noindent The gauge symmetry, 
\begin{equation}\label{gauge-symmetry}
    (\Psi(x), A(x)) \mapsto ( e^{i\eta(x)}\Psi(x),   A(x) + \nabla\eta(x)),\qquad \forall \eta \in C^2(\R^2, \R);
\end{equation}
The 
translation symmetry, 
\begin{equation}\label{translation-symmetry}
    (\Psi(x), A(x)) \mapsto (\Psi(x + t), A(x + t)),\qquad \forall t \in \R^2;
\end{equation}
The rotation and reflection symmetry, 
\begin{equation}\label{rotation-reflection-symmetry}
    (\Psi(x), A(x)) \mapsto (\Psi(R^{-1}x),  RA(R^{-1}x)),\qquad \forall R \in O(2) .
\end{equation}

An important role in our analysis is played by the reflections symmetry.  Let the reflection operator $T^{\rm refl}$ be given as \begin{equation}\label{rotation-reflection-symmetry}
T^{\rm refl}:    (\Psi(x), A(x)) \mapsto (\Psi(- x),  - A(- x)) .
\end{equation} We say that a state $(\Psi, A)$ is {\it even (reflection symmetric)} iff 
 \begin{align} \label{even} T^{\rm refl}(\Psi, A) = (\Psi, A).
\end{align} 

The reflections symmetry of the GLE equations implies that we can restrict the class of solutions to even ones.
In what follows we always assume that solutions $(\Psi, A)$ are even.

\medskip

\textbf{2.2 Elementary solutions.} There are two immediate solutions to the Ginzburg-Landau equations that are homogeneous in $\Psi$. These are the perfect superconductor
solution where $\Psi_S \equiv 1$ and $A_S \equiv 0$, and the \textit{normal} (or non-superconducting) solution where $\Psi_N = 0$ and $A_N$ is such that
$\Curl A_N =: b$ is constant.
(We see that the perfect superconductor is a solution only when the magnetic field is absent. On the other hand, there is a normal solution,  $(\Psi_N = 0,\ A_N,\ \Curl A_N =$ constant), for any constant magnetic field.)



\section{Lattice equivariant states}\label{sec:lattice states}

\textbf{3.1 Periodicity.} Our focus in this paper is on states $(\Psi, A)$ defined on all of $\R^2$, but whose physical properties, the density of superconducting pairs of electrons, $n_s := |\Psi|^2$, the magnetic field, $B := \Curl A$, and the current density, $J := \Im(\bar{\Psi}\Cov{A}\Psi)$, are doubly-periodic with respect
to some lattice $\cL$. We call such states $\cL-$\emph{lattice states}.

One can show that  a state $(\Psi, A) \in H^1_{\textrm{loc}}(\R^2;\C) \times H^1_{\textrm{loc}}(\R^2;\R^2)$ is a $\mathcal{L}$-lattice state if and only if translation by an element of the lattice results in a gauge transformation
    of the state, that is, for each $t \in \mathcal{L}$, there exists a function $g_t \in H^2_{loc}(\R^2;\R)$
    such that 
   \begin{equation}\label{equiv-cond} \Psi(x + t) = e^{ig_t(x)}\Psi(x)\ \mbox{and}\ A(x+t) = A(x) + \nabla g_t(x),\ \forall t\in \lat,\end{equation} almost everywhere. States satisfying \eqref{equiv-cond} will be called {\it($\lat-$) equivariant (vortex) states}.

It is clear that the gauge, translation, and rotation symmetries of the Ginzburg-Landau equations map lattice states to
lattice states. In the case of the gauge and translation symmetries, the lattice with respect to which the solution is
periodic does not change, whereas with the rotation symmetry, the lattice is rotated as well. It is a simple calculation
to verify that the magnetic flux per cell of solutions is also preserved under the action of these symmetries.

Note that $(\Psi, A)$ is defined by its restriction to a single cell and can be reconstructed from this restriction by lattice translations.

\textbf{3.2 Flux quantization.}  The important property of lattice states is that the magnetic flux through a lattice cell is quantized, 
\begin{equation}\label{eq:flux-per-cell}
    \int_{\Omega^\Lat} \Curl A = 2\pi n
\end{equation}
for some integer $n$, with $\Omega^\Lat$ any fundamental cell of the lattice.  This 
implies that
\begin{equation}\label{Ombrel}
	|\Omega^\Lat| = \frac{2\pi n}{b},
\end{equation}
where $b$ is the average magnetic flux per lattice cell, $b := \frac{1}{|\Omega^\Lat|} \int_{\Omega^\Lat} \Curl A$.

Indeed, if $|\Psi| > 0$ on the boundary of the cell, we can write 
$\Psi = |\Psi|e^{i\theta}$ and $0 \leq \theta < 2\pi$. The periodicity of $n_s$ and $J$ ensure the periodicity of $\nabla\theta - A$ and therefore by Green's theorem, $\int_\Omega \Curl A = \oint_{\partial\Omega} A = \oint_{\partial\Omega} \nabla\theta$ and this function is equal to $2\pi n$ since $\Psi$ is single-valued.

Equation \eqref{eq:flux-per-cell} then imposes a condition on the area of a cell, namely, \eqref{Ombrel}.

\textbf{3.3 Lattice shape.} 
We  identify $\R^2$ with $\C$, via the map $(x_1, x_2)\ra x_1+i x_2$,  and, applying a rotation, if necessary, bring any lattice $\LAT$ to the form  
 \begin{equation}\label{LATom}\LAT_\om=r  (\Z+\tau\Z),\end{equation}
where $\om=(\tau, r)$, $r>0,\ \tau\in \C$, $\Im\tau > 0$, which we assume from now on. 
If $\LAT_\om$ satisfies \eqref{Ombrel}, then $r = \sqrt{\frac{2\pi n}{b\im \tau}}$.
Furthermore, we introduce the normalized lattice
   \begin{equation}\label{LATtau}\LAT^\tau:=\sqrt{\frac{2\pi }{\im\tau} }  (\Z+\tau\Z) 
     \end{equation} 
   and let   $\Omega^\tau$ stand for an elementary cell of the lattice $\LAT^\tau$.  
    We note that $|\Omega^\tau| = 2\pi $. 

Since the action the modular group  $SL(2, \Z)$ on $\C$ does not change the lattice, but in general maps one basis, and therefore $\tau,\  \im\tau>0$, into another (see Supplement I of \cite{Sig}), the modular invariance of $\g( \tau)$ means that it depends only on the lattice and not its basis. Thus, it suffices to consider  $\tau$ in  the fundamental domain,  $\bH/SL(2, \Z)$, of   $SL(2, \Z)$ acting on the Poincar\'e half plane $\bH:=\{\tau\in \C: \im \tau >0\}$. 
The fundamental domain $\bH/SL(2, \Z)$ is given explicitly as
     \begin{align}\label{fund-domSL2Z} 
    \bH/SL(2, \Z)=\{\tau\in \C: \Im\tau > 0,\ |\tau| \geq 1,\ -\frac{1}{2} < \Re\tau \leq \frac{1}{2} \} . 
\end{align}



\section{Fixing the gauge and rescaling} \label{sec:rescaling}

In this section we fix the gauge for solutions,  $(\Psi, A)$, of \eqref{GLE} and then rescale them to eliminate the dependence of the size of the lattice on $b$. Our space will then depend only on the number of quanta of flux and the shape of the lattice.

\textbf{4.1 Fixing the gauge.}
The gauge symmetry allows one to fix solutions to be of a desired form.
Let $A^b (x) = \frac{b}{2} J x \equiv  \frac{b}{2} x^\perp$, where $x^\perp = Jx = (-x_2, x_1)$ and  $J$ is the symplectic matrix   
\begin{equation*}  J = \left( \begin{array}{cc} 0 & -1 \\ 1 & 0 \end{array} \right).\end{equation*}
     We will use the following preposition, first used by \cite{Odeh} and proved in \cite{Takac} (an alternate proof is given in in Appendix A of \cite{TS2}). 

\begin{proposition}
    \label{thm:fix-gauge}
    Let $(\Psi', A')$ be an $\mathcal{L}$-equivariant state, and let $b$ be the average magnetic flux per cell.
    Then there is a $\mathcal{L}$-equivariant state  $(\Psi, A)$, 
    that is gauge-equivalent to $(\Psi', A')$, such that
    \begin{enumerate}
    \item[(i)] 
  $\Psi(x + s) = e^{i(\frac{b}{2}x\cdot J s+c_s)}\Psi(x)$ and   $A(x + s) = A(x) +\frac{b}{2} J  s$ for all $s\in \LAT$; 
    \item[(ii)] 
  $ 
  \int_\Omega (A-A^b) = 0,\ \Div A = 0.$
  \end{enumerate}
Here  $c_s$ satisfies the condition 
$c_{s+t} - c_s - c_t - \frac{1}{2} b s \wedge t \in 2\pi\Z. $ 
\end{proposition}


\textbf{4.2 Rescaling.}
Let $\tau\in \C$, $\Im\tau > 0$ and $r = \sqrt{\frac{2\pi n}{b\im \tau}}$.  We define the rescaled fields $(\psi, a)$ as
\begin{align}\label{resc}
    (\psi(x), a(x)) := ( r' \Psi(r' x), r 'A(r' x) ),\ r':= r/ \sqrt{\frac{2\pi }{\im\tau} }=\sqrt{\frac{n}{b}}.
\end{align}
Let $\LATo$ and $\LATt$ be the lattices defined in \eqref{LATom} and \eqref{LATtau}.
    We summarize the effects of the rescaling above:

    \begin{enumerate}[(A)]

    \item $\Psi$ and $A$ solve the Ginzburg-Landau equations if and only if $\psi$ and $a$ solve
            \begin{subequations} \label{rGL}
            \begin{align} \label{rGLpsi}
                &(-\Delta_{a}  - \lambda) \psi = -\kappa^2 |\psi|^2\psi,\  \lambda = \kappa^2 n/ b,\\           
                 \label{rGLA}
              &  \Curl^*\Curl a = \Im(\bar{\psi}\Cov{a}\psi).
            \end{align} \end{subequations}

    \item   \label{reduced-gauge-form}
  $(\Psi, A)$ is  a $\LATo$-equivariant state iff $(\psi, a)$ is a $\mathcal{L}^\tau$-equivariant state. Moreover,  if $(\Psi, A)$ is of the form described in Proposition \ref{thm:fix-gauge}, then   $(\psi, a)$  satisfies 
           \begin{align}            \label{gaugeperiod-resc}  &\psi(x + t) = e^{i\frac{ n}{2}x\cdot J t+i c_t }\psi(x),\ a(x + t) = a(x)+ \frac{ n}{2} J t,\ \forall t \in \LAT^\tau\\   
            \label{aver-resc}   &    \Div  a = 0,\      \int_{\Omega^\tau} (a - a^n) = 0,\ \text{  where } 
               a^n(x) := \frac{ n}{2} J x,             \end{align}
              and  
              $c_t$, which satisfies the condition 
              \begin{align}            \label{cs-cond}c_{s+t} - c_s - c_t - \frac{1}{2} n s \wedge t \in 2\pi\Z. \end{align}  
       \item   $\frac{r^2}{|\Omega^\tau|}E_{\Omega^\tau}(\Psi,A) = \mathcal{E}_{\lambda}(\psi,\al)$,  where  $ a = a^n + \alpha,\  \mbox{with}\ a^n(x) := \frac{ n}{2} J x,$ $\lambda = \kappa^2 {r'}^2 = \kappa^2 \frac{n}{b}$ and
        \begin{equation}     \label{rEnergy}
            \mathcal{E}_\lambda(\psi, \alpha) = \frac{1}{|\Omega^\tau|} \int_{\Omega^\tau}\left( |\Cov{a}\psi|^2 + |\Curl a|^2
                        + \frac{\kappa^2}{2} ( |\psi|^2 - \frac{\lambda}{\kappa^2} )^2\right) \,dx.
        \end{equation}
                         \end{enumerate}
Our problem then is: for each $n = 1,2,\ldots$,  find $(\psi, a)$, 
solve the rescaled Ginzburg-Landau equations \eqref{rGL} and satisfying \eqref{gaugeperiod-resc}. 

In what follows, the parameter $\tau$ is fixed and, to simplify the notation, we {\it omit the superindex} $\tau$ at $\LAT^\tau$ and $\Omega^\tau$ and write simply $\LAT$ and $\Omega$. 


\section{The linear problem} 
\label{sec:operators}

In this section we consider the linearization of \eqref{rGL} satisfying \eqref{gaugeperiod-resc} on the normal solution  $(0, a^n),\  \mbox{with, recall,}\   a^n(x) := \frac{ n}{2} J x$. 
 This leads to the linear problem:
\begin{align} \label{lin-probl}-\Delta_{a^n} \psi_{0} = \lam  \psi_{0},\end{align} for $\psi_0$ satisfying the gauge - periodic boundary condition (see \eqref{gaugeperiod-resc})
 \begin{align}            \label{gaugeperiod-psi}\psi_0(x + t) = e^{i(\frac{  n}{2}x\cdot Jt+ c_t)}\psi_0(x),\ \forall t \in \LAT.\end{align}
Our goal is to prove the following 
  \begin{proposition}\label{prop:Landau-ham-spec}  The operator $-\Delta_{a^n}$ is self-adjoint on its natural domain and its spectrum is given by
  	\begin{equation}\label{spec-Landau}	\sigma(-\Delta_{a^n}) = \{\, (2m + 1) n : m = 0, 1, 2, \ldots \,\},
	\end{equation}
	with each eigenvalue is of the multiplicity $n$. Moreover,  
	\begin{align} \label{Vn-space} \Null (-\Delta_{a^n} - n) =e^{\frac{in}{2}x_2(x_1 + ix_2) } V_n,\end{align} 
where $V_n$ is spanned by functions of the form (below $z= (x_1+i x_2)/ \sqrt{\frac{2\pi}{\im\tau} }$)
     \begin{align} 
  \label{theta-repr} &\theta (z, \tau) := \sum_{m=-\infty}^{\infty} c_m e^{i2\pi m z},\  c_{m + n} = e^{-in\pi z} e^{i2m\pi\tau} c_m.  \end{align}
Such functions are determined entirely by the values of $c_0,\ldots,c_{n-1}$ and therefore form an $n$-dimensional vector space.   \end{proposition}
  \begin{proof} 
 The self-adjointness of the operator $-\Delta_{a^n}$ is well-known. To find its spectrum, we  introduce  the complexified covariant derivatives (harmonic oscillator annihilation and creation  operators), $\bar\p_{a^n} $ and $\bar\p_{a^n}^*=- \p_{a^n}$, with 
    \begin{equation}
        \bar\p_{a^n}  := (\nabla_{a^n})_1 + i(\nabla_{a^n})_2 =\partial_{x_1} + i\partial_{x_2} + \frac{1}{2} n (x_1 + i  x_2).
    \end{equation}
    One can verify that these operators satisfy the following relations:
     \begin{align} \label{commut-rel}[\bar\p_{a^n}, (\bar\p_{a^n})^*] &= 2\Curl a^n =2 n;\\
  \label{Landau-cr-annih}    -\Delta_{a^n} -  n &= (\bar\p_{a^n})^*\bar\p_{a^n}.  \end{align}
    As for the harmonic oscillator (see for example \cite{GS2}), this gives explicit information  about the spectrum of $-\Delta_{a^n}$, namely \eqref{spec-Landau}, with each eigenvalue is of the same multiplicity. 
    Furthermore, the above properties imply
    \begin{equation} \label{nullL'}
        \Null (-\Delta_{a^n} - n) = \Null \bar\p_{a^n}.
    \end{equation}

     We find $\Null \bar\p_{a^n}$.
       A simple calculation gives the following operator equation 
\begin{equation*}
  e^{-\frac{ n}{2}(i x_1 x_2-x_2^2)}\bar\p_{a^n} e^{\frac{ n}{2}(i x_1 x_2-x_2^2)} = \partial_{x_1} + i\partial_{x_2}.
\end{equation*}
(The transformation on the l.h.s. is highly non-unique.)    
                 This immediately proves that 
                  \begin{equation}\label{lin-probl'} \bar\p_{a^n} \psi = 0,\end{equation}
                  if and only if $\theta = e^{-\frac{ n}{2}(i x_1 x_2-x_2^2)}\psi$ satisfies $(\partial_{x_1} + i\partial_{x_2})\theta = 0$.
        We now identify $x \in \R^2$ with $z = x_1 + ix_2 \in \C$ and see that this means that $\theta$ is analytic and
                \begin{equation}\label{psi-nullL}    \psi \left( x \right) =   e^{ -\frac{\pi n}{2\im\tau} (|z|^2-z^2) } \theta(z, \tau),\ z = (x_1+ i x_2)/ \sqrt{\frac{2\pi }{\im\tau} }.
        \end{equation}
where we display the dependence of $\theta$ on $\tau$. 
The quasiperiodicity of $\psi$ transfers to $\theta$ as follows
            \begin{equation*}
                \theta(z + 1, \tau) = \theta(z, \tau), \qquad
                \theta(z + \tau, \tau) =  e^{ -2\pi inz } e^{ -in\pi\tau  } \theta(z, \tau).
            \end{equation*}

        The first relation ensures that $\theta$ have a absolutely convergent Fourier expansion of the form
        $    \theta(z, \tau) = \sum_{m=-\infty}^{\infty} c_m e^{2\pi m iz}.$ 
        The second relation, on the other hand, leads to relation for the coefficients of the expansion:        
           $ c_{m + n} = e^{-in\pi z} e^{i2m\pi\tau} c_m$, 
 which together with the previous statement implies \eqref{theta-repr}.              
    \end{proof}
  Next, we claim that the solution \eqref{psi-nullL} satisfies 
  \begin{align}  \label{psi-parity} \psi (x) = \psi (- x). \end{align}
By \eqref{psi-nullL}, it suffices to show that $\theta (z) = \theta (- z)$.  We show this for $n=1$. Denote the corresponding $\theta$ by $\theta (z, \tau)$. Iterating the recursive relation for the coefficients in \eqref{theta-repr}, we obtain the following
  standard representation 
  for the theta function    
\begin{align}  \label{theta-series} &\theta (z, \tau) =   \sum_{m=-\infty}^{\infty} e^{2\pi i  ( \frac12 m^2\tau + m z)}. \end{align} 
We observe that $\theta (- z, \tau)=\theta (z, \tau)$ and therefore $\psi_0 ( - x)=\psi_0 ( x)$. Indeed, using the expression \eqref{theta-series}, we find, after changing $m$ to $-m'$, we find
\[\theta (- z, \tau)=   \sum_{m=-\infty}^{\infty} e^{2\pi i  ( \frac12 m^2\tau - m z)}
=   \sum_{m'=-\infty}^{\infty} e^{2\pi i  ( \frac12 m'^2\tau + m' z)} =\theta (z, \tau).\]


\section{Setup of the bifurcation problem} \label{sec:setup-bifurc}

In this section we reformulate the Ginzburg-Landau equations as a bifurcation problem. 
 We pick the fundamental cell with the center at the origin so that it is invariant under the map $x\ra -x$.  
 
We write $a = a^n + \alpha$ and substitute this into \eqref{rGL} to obtain
    \begin{equation} \label{psia-eqs}  
        (L^n - \lambda)\psi  =- h(\psi, \al),\ \quad
        M \al = J(\psi, \al) ,
    \end{equation}
where $h(\psi, \al):= 2i\alpha\cdot\nabla_{a^n}\psi 
        + |\alpha|^2\psi + \kappa^2|\psi|^2\psi$ and $ J(\psi, \al) :=\Im(\bar{\psi}\Cov{a^n+\al}\psi )$ and
        \begin{equation}\label{Ln}
    L^n := -\Delta_{a^n}
\mbox{ and }
    M := \Curl^*\Curl.
\end{equation}
 The pair $(\psi, \al)$ satisfies the conditions \eqref{gaugeperiod-resc} - \eqref{aver-resc}, 
  with  $ a = a^n + \alpha,\  a^n(x) := \frac{ n}{2} J x,$  which we reproduce here
  \begin{align}            \label{gaugeperiod-psi'}  &\psi(x + t) = e^{i(\frac{ n}{2}x\cdot J t+c_t) }\psi(x),\\ 
   \label{alpha-per'} &\alpha(x + t) = \alpha(x)\ 
  {\rm and}\ \divv \alpha = 0,\ \int_{\Omega} \alpha = 0, \end{align}
where  $t  \in \LAT$ 
 and $c_t$ satisfies the condition  \eqref{cs-cond}.      We take $c_t=0$ on the basis vectors $t=\sqrt{\frac{2\pi }{\im\tau} }, \sqrt{\frac{2\pi }{\im\tau} } \tau$ (see \eqref{LATtau}). Then the relation  \eqref{cs-cond} gives
   \begin{align}   \label{cs-express}c_{s} =\pi n p q,\ \text{ for }\ s=\sqrt{\frac{2\pi }{\im\tau} } (p+ q\tau), p, q\in \Z, \end{align}
which we assume in what follows.

We consider \eqref{psia-eqs} on the space $\sH_{n}^2\times\HA{2}{\Div,0}$, where $\sH^s_n$ and $\HAs{\Div,0}{s}$ are the Sobolev spaces of order $s$ associated with the $L^2$-spaces 
 \[\Lpsi{2}{n}:=\{\psi\in L^2(\R^2, \C): 
 \psi$ satisfies$\ \psi ( - x)=\psi ( x)\  
 $and \eqref{gaugeperiod-psi'}$\},\] 
 \[\LA{p}{\Div,0}:=\{\al\in L^2(\R^2, \R^2)\ |\ \al \text{ satisfies$\  a ( - x)=- a ( x)\  
 $and  \eqref{alpha-per'}} 
 \}, \]
  where $\divv \al$ is understood in the distributional sense, with the inner products of $L^2(\Om, \C)$ and $L^2(\Om, \R^2)$, i.e. $\int \bar{\psi}\psi'$ and $\int \alpha\cdot \alpha'$. 
 
We define $L^n$ and $M$ on the spaces  $\Lpsi{2}{n}$ and $\LA{p}{\Div,0}$, with the domains  $\sH^{2}_{n}$ and $\HA{2}{\Div,0}$, respectively. The properties of $L^n$ where described in Proposition \ref{prop:Landau-ham-spec}. The properties  of $M$ are summarized as in the following proposition

\begin{proposition}\label{thm:M-spec}
    $M$ is a strictly positive operator on $\LA{p}{\Div,0}$ with the domain  $\HA{2}{\Div,0}$ and with purely discrete spectrum.
\end{proposition}
The fact that $M$ is positive follows immediately from its definition. We note that its being strictly positive is the result of
    restricting its domain to elements having the divergence and mean zero.

\begin{proposition}\label{prop:reconstr}  Assume $(\lambda, \psi, \alpha)$ is a solution of the system \eqref{psia-eqs'} satisfying \eqref{gaugeperiod-psi'}-\eqref{alpha-per'}. 
Then $\divv J (\psi, \al)=0$ and  $\lan J (\psi, \al)\ran =0$ and $(\lambda, \psi, \alpha)$ solves the system \eqref{psia-eqs}.  
   \end{proposition}
\begin{proof} Let $P'$ be the orthogonal projection onto the divergence free, mean zero vector fields ($P'=\frac{1}{-\Delta}\curl^*\curl$). We introduce the new system
    \begin{equation} \label{psia-eqs'} 
        (L^n - \lambda)\psi  + h(\psi, \al)=0,\ \quad
        M \al  -P' J(\psi, \al) =0,
    \end{equation}
where we left the first equation unchanged and in the second equation we introduced the projection $P'$ . 

We rewrite \eqref{psia-eqs'} as a single equation
\begin{equation}\label{Feq0} F(\lambda, \psi, \al) = 0,
\end{equation}
 where the map
$F : \R \times \sH^{2}_{n} \times \HAs{\Div,0}{2}\to \Lpsi{2}{n}\times \LA{p}{\Div,0}$ is defined by the l.h.s. of \eqref{psia-eqs'} as 
\begin{equation}\label{F}
    F(\lambda, u) = A_\lam u +f(u) .
\end{equation}
Here $u:=(\psi, \al),\ A_\lam:=\diag (L^n - \lambda, M)$ and 
\begin{equation} \label{f}
 f(u):= ( h(u), - P' J(u) ). 
 \end{equation}

For a map $F(\lambda, u),\ u=(\psi, \al)$, we denote by $\p_\psi F(\lambda, u)/\p_u F(\lambda, u)$ its G\^ateaux derivative in $\psi/u$. 
Furthermore, we use the obvious notation $F=(F_1, F_2)$. For $f=(f_1, f_2)$, we introduce the gauge transformation as $T_\del f= (e^{i\delta} f_1, f_2)$. The following proposition lists some properties of $F$.
\begin{proposition}\label{prop:F-propert}
    \hfill
    \begin{enumerate}[(a)]
    \item $F$ is analytic as a map between real Banach spaces,
    \item for all $\lambda$, $F(\lambda, 0) = 0$,
    \item for all $\lambda$, $\p_u F(\lambda, 0)=A_\lam, $ 
    \item for all $\delta \in \R$, $F(\lambda, T_{\delta}u) = T_\del  F(\lambda, u)$.
	\item \label{psiF-real} for all $u$ (resp. $\psi$), $\langle u, F(\lambda, u) \rangle \in \R$ (resp. $\langle \psi, F_1(\lambda, u) \rangle \in \R$).
    \end{enumerate}
\end{proposition}
\begin{proof}
    The first property follows from the definition of $F$. 
    (b) through (d) are straightforward calculations. For (e), since $\langle u, F \rangle =\langle \psi, F_1 \rangle +\langle \alpha, F_2 \rangle$ and $\langle \alpha, F_2(\lambda, u) \rangle$ is real, the statements $\langle u, F(\lambda, u) \rangle \in \R$ and $\langle \psi, F_1(\lambda, u) \rangle \in \R$ are equivalent. Now, we calculate that
   	\begin{align*}
		\langle \psi, F_1(\lambda, \psi, \al) \rangle
		&= \langle \psi, (L^n - \lambda)\psi \rangle + 2i\int_{\Omega} \bar{\psi}\al \cdot\nabla\psi\\
			&+ 2\int_{\Omega} (\alpha \cdot a^n)|\psi|^2
			+ \int_{\Omega} |\al |^2 |\psi|^2
			+ \kappa^2 \int_{\Omega} |\psi|^4.
	\end{align*}
	The final three terms are clearly real and so is the first because $L^n - \lambda$ is self-adjoint. For the second term we integrate by parts and  use the fact that the boundary terms vanish due to the periodicity of the integrand to see that
	\begin{equation*}
\Im 2i\int_{\Omega}  \phi \alpha \cdot\nabla\bar{\psi}			= \int_{\Omega}\al \cdot (\bar{\psi} \nabla\psi + \psi  \nabla\bar{\psi} )
			= -\int_{\Omega} (\divv \al) |\psi|^2=0,			\end{equation*}
	where we have used that $\Div \al = 0$. Thus this term is also real and (e) is established.
\end{proof}
 We return to the proof of Proposition \ref{prop:reconstr}. Assume 
 $\chi\in H^1_{\rm loc}$ and is $\lat-$periodic (we say, $\chi \in H^1_{\rm per} $). Following \cite{TS}, we differentiate the equation $\cE_\lam(e^{i s\chi}\psi, \al+s\nabla\chi)=\cE_\lam(\psi, \al) $, w.r.to $s$ at $s=0$, use that $\curl\n \chi=0$  and integrate by parts, to obtain
\begin{align} 
 \Re\lan -\Delta_{a^{n} + \alpha}\psi  + \kappa^2(|\psi|^2 - 1)\psi, & i \chi\psi\ran\notag \\ 
\label{E-deriv}& + \lan 
J (\psi, \al), \n \chi\ran=0.\end{align} 
\DETAILS{\[ \p_\psi \cE_\lam(\psi, \alpha)i\chi\psi+\p_\alpha \cE_\lam(\psi, \alpha)\nabla\chi=0.\] Here $\p_\psi \cE_\lam(\psi, \alpha)$ and $\p_\alpha \cE_\lam(\psi, \alpha)$ are the G\^ateaux derivatives of $\cE_\lam(\psi, \alpha)$ w.r.to $\psi$ and $\alpha$. Since the first equation in \eqref{psia-eqs'} is equivalent to $ \p_\psi \cE_\lam(\psi, \al)=0$ and since  $\p_\alpha \cE_\lam(\psi, \alpha)\nabla\chi= 
  \int_{\Omega^\tau} (M \al -J (\psi, \al))\cdot\nabla\chi,$ we find that \[  \int_{\Omega^\tau} (M\al -J (\psi, \al))\cdot\nabla\chi=0.\]
Taking $\chi \in H^1(\Omega, \R)$ periodic and integrating by parts, we derive that $\int_{\Omega} \divv J (\psi, \al)\chi=0$. Since the last equation holds for any periodic $\chi \in H^1(\Omega, \R)$, we conclude that $\divv J (\psi, \al)=0$.}
 (Due to conditions \eqref{gaugeperiod-psi'} - \eqref{alpha-per'} and the $\lat-$periodicity of $\chi$, 
 there are no boundary terms.) 
This, together with the first equation in \eqref{psia-eqs'},  
implies \begin{align}\label{J-relat} \lan J (\psi, \al), \n\chi\ran=0.\end{align} 
 Since the last equation holds for any  $\chi \in H^1_{\rm per} $, we conclude that $\divv J (\psi, \al)=0$.

Furthermore, since for our class of solutions $\psi$ is even and $a$, odd, we conclude that $J(\psi, a)$ is odd under reflections and therefore $\langle J(\psi, a) \rangle = 0$. 
\end{proof} 
In Sections \ref{sec:reduction} - \ref{sec:bifurc-point-sym} 
 we solve the system \eqref{psia-eqs'}, subject to the conditions \eqref{gaugeperiod-psi'}-\eqref{alpha-per'}. 
 
\section{Reduction to a finite-dimensional problem} \label{sec:reduction}

In this section we reduce the problem of solving  the equation $F(\lambda, u) = 0$, were $F$ is given in \eqref{F} and $u:=(\psi, \al) $, to a finite dimensional problem. We address the latter in the next section.
We use the standard method of Lyapunov-Schmidt reduction.

We do the reduction in the generality we need later on.  Let  $X=X' \times X''$ and $Y=Y' \times Y''$ be closed subspaces of  $\sH^{2}_{n}\times \HA{2}{\Div,0}$ and $\Lpsi{2}{n}\times \LA{p}{\Div,0}$, respectively, s.t.  
 \begin{align}\label{XY-cond} X\subset Y,\ \text{densely,\   and }\ F : \R \times X \to Y \text{ and is } C^2.\end{align} 

 Recall that the operator $A_{\lam}:=\diag (L^n - \lam, M)$ is introduced after the equation \eqref{F}. Since $A_{\lam}= dF(\lambda, 0)$, it maps $X$ into $Y$. We 
 let $K = \Null_{X } A_{n}\subset X$. 
    
       We let $P$ be the orthogonal projection in $Y$ onto $K$ and let $\bar P := I - P$.  Since $ n$ is an isolated eigenvalue of $A_{ n}$, $P$ can be explicitly given as the  Riesz projection,
    \begin{equation}
        P := -\frac{1}{2\pi i} \oint_\gamma (A_{ n} - z)^{-1} \,dz,
    \end{equation}
    where $\gamma \subseteq \C$ is a contour around $0$ that contains no other points of the spectrum of $A_{ n}$. 
    
    Writing $u=v +w$, where $v = P u$ and $w = \bar P u$, we see that the equation $F(\lambda,u) = 0$ is therefore equivalent to the pair of equations
    \begin{align}
        \label{LS:eqn1} &P F(\lambda, v + w) = 0, \\
        \label{LS:eqn2} &\bar P F(\lambda, v + w) = 0.
    \end{align}
   We will now solve \eqref{LS:eqn2} for $w = \bar P u$ in terms of $\lambda$ and $v = P u$. 
   \begin{lemma}\label{w-est} There is a neighbourhood, $U\subset \R \times K$, of $(n, 0)$,
    such that for any $(\lambda, v)$ in that neighbourhood, Eq \eqref{LS:eqn2} has a unique solution $w = w(\lambda, v)$. This solution $w(\lambda, v)=(w_1, w_2)$ satisfies 
  \begin{align}\label{w-prop1}
&  w(\lambda, v)  \ \mbox{real-analytic in}\    (\lambda, v),\\ 
   \label{w-prop2}
    &   \|\p_\lam^m w_i\|=O(\|v\|^{4-i}),\ i=1, 2,\ m=0, 1, \end{align}
where the norms are in the space $\sH^{2}_{n}$.    \end{lemma}
    \begin{proof} We introduce the map
    $G : \R \times K \times \bar X \to \bar Y$, where  $\bar X := \bar P X=X\ominus K$ and $\bar Y :=\bar P Y= Y \ominus K$, defined by \[G(\lambda, v, w) =\bar P F(\lambda, v + w).\] 
 It has the following properties (a) $G$ is $C^2$; (b) $G(\lambda, 0, 0)=0$ $\forall\lambda$; (c) $d_w G(\lambda, 0, 0)$  is invertible for $\lambda =n$.      Applying the Implicit Function Theorem
    to $G=0$, we obtain a function $w : \R \times K \to \bar X$, defined on a neighbourhood of $(n, 0)$,
    such that $w = w(\lambda, v)$ is a unique solution to $G(\lambda, v, w) = 0$, for $(\lambda, v)$ in that neighbourhood. This proves the first statement.
    
    By the implicit function theorem and the analyticity of $F$, the  solution has the  property \eqref{w-prop1}.

    By  \eqref{F} and the fact that product of  $\sH^{2}_{n},  \HA{2}{\Div,0}$  functions is again a  $\sH^{2}_{n},  \HA{2}{\Div,0}$ function (and the norms are bounded correspondingly), implies that  
    \[\| h(u) \|_{H^2} \ls \| u \|_{H^2}^3$\ and\  $\| J(u) \|_{H^2} \ls \| u \|_{H^2}^2,\] where $h$ and $J$ defined after the equation \eqref{Ln}.  
    Using the definition \eqref{F}, we can rewrite  \eqref{LS:eqn2} as 
\begin{equation}\label{w-eq}
  A_\lam w = - \bar P f(\lambda, u).
\end{equation}
Since by Proposition \ref{prop:Landau-ham-spec}, $A_\lam^\perp:= \bar P A_\lam  \bar P\big |_{\Ran  \bar P}$ is invertible for $\lam$ close to $n k$, with the uniformly bounded inverse and since $A_\lam$ is diagonal and $f$ is of the form \eqref{f}, we conclude that   $\|w_1\|\ls \| h(u) \|_{H^2} \ls \| u \|_{H^2}^3$ and $\|w_2\|\ls \| J(u) \|_{H^2} \ls \| u \|_{H^2}^2$. Recalling that $u=v+w$, this gives the second relation in \eqref{w-prop2}. The first relation in \eqref{w-prop2} is proven similarly.   \end{proof}

    We substitute the solution $w = w(\lambda, v)$ into \eqref{LS:eqn1} and see that the latter equation 
     in a neighbourhood of $(n, 0)$ is equivalent to the equation  (the \emph{bifurcation equation})
           \begin{equation}  \label{bif-eqn}
        \gamma(\lambda, v):= PF(\lambda, v + w(\lambda, v)) = 0.
    \end{equation}
    Note that 
     $\gamma : \R \times K \to K$. 
     We show that $w$ and $\gamma$ inherit the symmetry of the original equation:
    \begin{lemma}\label{lem:w-gam-gaugeinv}
        For every $\delta \in \R$, $w(\lambda, e^{i\delta}v) = T_\del w(\lambda, v)$ and $\gamma(\lambda, e^{i\delta}v) = e^{i\delta} \gamma(\lambda, v)$.
    \end{lemma}
    \begin{proof}
        We first check that $w(\lambda, e^{i\delta}v) = T_\del w(\lambda, v)$. We note that by definition of $w$,
         \begin{align*}G(\lambda,  e^{i\delta} v, w(\lambda, e^{i\delta} v)) = 0, \end{align*} but by the symmetry of $F$, we also have
        $G(\lambda,  e^{i\delta} v,  e^{i\delta} w(\lambda, v)) = T_\del  G(\lambda,v, w(\lambda,v)) = 0$. The uniqueness of $w$
        then implies that $w(\lambda, e^{i\delta} v) = T_\del w(\lambda, v)$.
        Using that $e^{i\delta} v=T_\del v$, we can now verify that
        \begin{align*}
            \gamma(\lambda, e^{i\delta} v) = PF(\lambda, e^{i\delta} v + w(\lambda,  e^{i\delta} v))&= PF(\lambda, T_\del (v + w(\lambda,  v)))\\
                &= PT_\del F(\lambda, v + w(\lambda,  v)). 
        \end{align*}
 Since  
$ P$ is of the form $P=P_1\oplus 0$, where $P_1$ acts  on the first component, we have $PT_\del F(\lambda, v + w(\lambda,  v))= e^{i\delta} PF(\lambda, v + w(\lambda, v) ) = e^{i\delta}\gamma(\lambda,v)$, which implies the second statement.   \end{proof}
  Thus we have shown the following
  \begin{corollary}\label{cor:reduct-fd}
  In a neighbourhood of $(n, 0)$ in $\R \times X$, $(\lambda, u),$ where $ u=(\psi, \al),$ solves Eqs \eqref{rGL} or \eqref{psia-eqs} 
    if and only if $(\lambda, v)$, with $v = P u$, solves
    \eqref{bif-eqn}. Moreover, the solution $u$ of \eqref{psia-eqs} can be reconstructed from the solution $v$ of \eqref{bif-eqn} according to the formula
     \begin{equation} \label{uvw}
     u =v+ w(\lambda, v).
    \end{equation} 
    \end{corollary}

Solving the bifurcation equation \eqref{bif-eqn} is a subtle problem. 
We do this in the next section assuming  $\dim_\C\Null_{X'} (L^{n}-n)=1$.

\section{Existence result assuming $\dim_\C\Null_{X'} (L^{n}-n)=1$} 
 \label{sec:bifurc-dimK=1} 

The main result of this section is the following theorem which gives a general, but conditional result. 

\begin{theorem} \label{thm:bif-thm-Kdim1}
Assume (i) $\LAT$ satisfies \eqref{LAT-cond}, (ii) \eqref{XY-cond} holds and (iii)  
\begin{align}            \label{Null-dim1-cond}  \dim_{\C}\Null_{X' } (L^{n}-n)=1. 
\end{align}   Then, 
for every $\tau$, 
    there exist $\epsilon > 0$ and a branch, $ (\lambda_s, \psi_s, \alpha_s)$, $s \in [0, \sqrt{\e})$, 
    of nontrivial solutions of  the rescaled Ginzburg-Landau equations \eqref{rGL}, unique modulo the global gauge symmetry  (apart from the trivial solution $(n, 0, a^n)$) in a sufficiently small neighbourhood of $(n, 0, a^n)$ in $\R \times X$, 
    and such that
    \begin{equation}\label{s-expansions}
    \begin{cases}
        \lambda_s =  n + g_\lambda(s^2), \\
        \psi_s = s\psi_0 + g_\psi(s^3), \\
        a_s = a^n + g_a(s^2),
    \end{cases}
    \end{equation}
    where $\psi_0$ is the solution of the problem \eqref{lin-probl} - \eqref{gaugeperiod-psi}, with $\lam=n $, 
   (normalized as $\lan|\psi_0|^2\ran =1$), 
    $ g_\psi$ is orthogonal to $\Null(L^n -  n)$,  $g_\lambda : [0,\epsilon) \to \R$, $g_\psi : [0,\epsilon) \to \sH^{2}_{n}$, and $g_\al : [0,\epsilon) \to  \HA{2}{\Div,0}$ 
     are real-analytic functions such that $g_\lambda(0) = 0$, $g_\psi(0) = 0$, and $g_\al(0) = 0$.  
    \end{theorem}
\begin{proof} The proof of this theorem is a slight modification of a standard result from bifurcation theory. 
Our first goal is to solve the equation \eqref{bif-eqn} for $\lam$.       By Proposition \ref{thm:M-spec}, we have 
   \begin{align}            \label{NullA-NullL}  \Null_{X } A_{n}= \Null_{X' } (L^{n}-n)\times \{0\}. 
\end{align}    
 This relation and  assumption \eqref{Null-dim1-cond} yield that    the projection $P$ is rank one 
 and therefore it can be written, for $u=( \phi, \beta)$, as 
    \begin{align}    \label{P}
        P u 
        = s v_0,\ \textrm{with}\ & s:=\frac{1}{\|\psi_0\|^2} \langle \psi_0, \phi \rangle,\ v_0:=  (\psi_0, 0),\\ \notag
        &\psi_0 \in \Null_{X' } (L^n - n),\ \|\psi_0\|=1. 
    \end{align}
Hence,          we can  write the map $\gamma$ in the bifurcation equation \eqref{bif-eqn} as $\gamma=\psi_0 \tilde \g$, where $\tilde \gamma : \R \times \C \to \C$ is given by
		\begin{equation} \label{gam}
		\tilde \gamma(\lambda, s) = \langle \psi_0, F_1(\lambda, s v_0 + w(\lambda, s v_0)) \rangle.
	\end{equation}
We now show that $\tilde \gamma(\lambda, s) \in \R$ for $s \in \R$.
Since the projection $\bar P$ is self-adjoint,  $\bar P w(\lambda, v) = w(\lambda, v), $ $w(\lambda, v)$ solves  $ \bar P F(\lambda, v + w)=0$ and $v=(s \psi_0, 0)=s  v_0,$ we have
	\begin{align*}
		\langle w(\lambda, s v_0), F(\lambda, s v_0+ w(\lambda, s v_0)) \rangle
		= \langle w(\lambda, s v_0), \bar P F(\lambda, s v_0 + w(\lambda, s v_0)) \rangle
		= 0.
	\end{align*}
	Therefore, for $s \neq 0$,
	\begin{align*}
		\langle \psi_0, F_1(\lambda, s v_0 + w(\lambda, s v_0)) \rangle
	&= s^{-1}\langle s  v_0, F(\lambda, s v_0 + w (\lambda, s v_0)) \rangle	\\ 
	&= s^{-1}\langle s  v_0 + w (\lambda, s v_0), F(\lambda, s v_0 + w (\lambda, s v_0)) \rangle,
	\end{align*}
	and this is real by property \eqref{psiF-real} of Proposition \ref{prop:F-propert} and the fact that the part $\langle  w_2 (\lambda, s v_0),$ $ F_2(\lambda, s v_0 + w (\lambda, s v_0)) \rangle$ of the inner product on the r.h.s.  is real. Next, by Lemma \ref{lem:w-gam-gaugeinv}, $\tilde\gamma(\lambda, s) = e^{i\arg s} \tilde\gamma(\lambda, |s|)$. Therefore it suffices to solve the equation
\begin{equation} \label{bif-eqn2} \gamma_0(\lambda, s) =0
\end{equation}
for the restriction $\gamma_0 : \R \times \R \to \R$ of the function $\tilde\gamma$ to $\R \times \R$, i.e., for real $s$. It follows from \eqref{w-prop2} that $w(\lambda, s v_0)=O(s^3)$, and therefore \eqref{bif-eqn2} has the trivial branch of solutions $s\equiv 0$ for all $\lam$. Hence we factorize $\gamma_0(\lambda, s)$ as $\gamma_0(\lambda, s) =s\gamma_1(\lambda, s)$, i.e., 
 we define the function
 \begin{align} \label{gam1}
		\gamma_1(\lambda, s) &:= s^{-1}\gamma_0(\lambda, s),\  \text{ if }\ s>0,\    \text{ and }\ =0   \text{ if }\ s=0,
	\end{align}
and solve the equation $\gamma_1(\lambda, s) =0$ for $\lam$. The definition of the function $\gamma_1(\lambda, s)$  and \eqref{w-prop1} imply that it has the following properties: $\gamma_1(\lambda, s)$ is real-analytic, $\gamma_1(\lambda, -s) = \gamma_1(\lambda, s)$, $\gamma_1 (1, 0)=0$ and, by \eqref{F} and \eqref{w-prop2}, $\p_\lam\gamma_1(1, 0)=-\|\psi_0\|^2\ne 0$. Hence by a standard application of  the Implicit Function Theorem,  there is $\epsilon > 0$ and a real-analytic function
$\widetilde{\phi}_\lambda : (-\sqrt{\epsilon}, \sqrt{\epsilon}) \to \R$ such that $\widetilde{\phi}_\lambda(0) = 1$ and $\lambda = \widetilde\phi_\lambda(|s|)$ solves the equation  $\gamma_1(\lambda, s) = 0$ with $|s| < \sqrt{\epsilon}$. 

We also note that because of the symmetry, $\widetilde{\phi}_\lambda(-|s|) = \widetilde{\phi}_\lambda(|s|)$, 
$\widetilde{\phi}_\lambda$ is an even real-analytic function, and therefore must in fact be a function solely of $s^2$. We therefore set $\phi_\lambda(s) = \widetilde{\phi}_\lambda(\sqrt{s})$ for $s \in [0, \e)$, and
so $\phi_\lambda$ is real-analytic.

Let  $w=(w_1, w_2)$.  We now define $g_\psi : [0, \epsilon) \to \sH^{2}_{1} $ and $g_a : [0, \epsilon) \to   \HA{2}{\Div,0}$ as 
    \begin{equation}
        g_\psi(s) = \begin{cases}
                        \frac{1}{\sqrt{s}} w_1(\phi_\lambda(s), \sqrt{s} v_0) &  s \neq 0\\                        0 & s = 0,
                    \end{cases}\ 
                    \qquad  g_a (s) = \begin{cases}
                         w_2(\phi_\lambda(s), \sqrt{s} v_0) &  s \neq 0\ \qquad  \\
                        0 & s = 0,
                    \end{cases}
    \end{equation}
By the definition,  $(s g_\psi(s^2), g_a(s^2)) = w(\widetilde{\phi}_\lambda(s), s v_0)$ for any $s \in [0, \sqrt{\e})$. 

Now, we know that there is a neighbourhood of $(k, 0, 0)$ in $\R \times X_n$, 
such that in this neighbourhood $F(\lambda, u) = 0$
if and only if $u=s v_0 +w(\lam, s v_0)$ and $\gamma(\lambda, s) = 0$. 
By taking a smaller neighbourhood if necessary, we have proven that
$F(\lambda, u) = 0$ in this neighbourhood if and only if either $s = 0$ or $\lambda = \phi_\lambda(s^2)$. If $s = 0$, we have
$u = s v_0 + w(\widetilde{\phi}_\lam(s), s v_0) = 0$, which gives the trivial solution.
In the other case, $u = s v_0 + w(\widetilde{\phi}_\lam(s), s v_0) = (s \psi_0 + sg_\psi(s^2), g_a (s^2))$. 

This gives the unique non-trivial solution, $(\phi, \beta)$, to the equation \eqref{Feq0}, or \eqref{psia-eqs}.  By Proposition \ref{prop:reconstr}, this solution is gauge equivalent to a solution, $(\psi, \al)$, to the equation \eqref{Feq0}, or \eqref{psia-eqs}. One can show easily that $(\psi, \al)$ satisfies the same estimates as $(\phi, \beta)$, so we keep the notation $g_\psi(s^2)$ and $g_a (s^2)$ for the corresponding terms for $(\psi, \al)$.

If we now also define $\lambda_s:=\phi_\lam(s)=: k n + g_\lam(s^2),\ \psi_s:=s \psi_0 (1 + g_\psi(s^2))$ and $a_s := a_0 + g_a(s^2)$, we see that the solution, $ (\lambda_s, \psi_s, a_s)$, we obtained,  solves  the rescaled Ginzburg-Landau equations \eqref{rGL} and is of the form \eqref{s-expansions}. \end{proof}

\section{Bifurcation theorem for $n=1$} 
 \label{sec:bifurc-n=1}

In this section we let $\LAT=\LAT_\om$ be an arbitrary 
 lattice satisfying  \eqref{LAT-cond} and take $n=1$. 
For $n=1$, we can take $X =\sH^{2}_{1}\times \HA{2}{\Div,0}$ and $Y = \Lpsi{2}{1}\times \LA{p}{\Div,0}$ (see the paragraph preceding \eqref{XY-cond}).     By Eq \eqref{Null-dim1-cond} and Proposition \ref{prop:Landau-ham-spec}, the space  $K = \Null A_{n}$ has the complex dimension $1$ and therefore Theorem \ref{thm:bif-thm-Kdim1} is applicable and gives Theorem \ref{thm:main-resultTS}, statements (I) and (II). An additional simple argument gives (III).

\section{Bifurcation theorem  for $n > 1$ and point symmetries} 
 \label{sec:bifurc-point-sym}

 As above, $n$ will denote the number of flux quanta through a fundamental cell of $\cL$. We want to prove the existence of Abrikosov lattices for $n \geq 2$. 
 The main notions of the section are as follows:
\begin{enumerate}
	\item Number of flux quanta, $n$.
	\item $\cL-$ irreducibility. We are interested in $\cL-$equivariant solutions which are not equivariant for any finer lattice. We call such solutions $\cL-$irreducible. 
	\item Multiplicity, which is defined as the dimension of the linear subspace, $\Null_{X } A_n$. The difficulty of bifurcation theory reduces considerably if the multiplicity is one. 	We call the corresponding solutions {\it simple}. 
\end{enumerate}

The former is achieved by  employing the symmetries of the lattice to reduce the dimension of $\Null_{X } A_{ n}$, more precisely, to find $X$ satisfying  \eqref{XY-cond} and $\dim_{\C}\Null_{X } A_{ n}=1.$

 In the next two subsections, we outline the general strategy of reducing multiplicity by the group symmetry and choose appropriate subgroups of the point group to impose as symmetry group. Then, in the following three subsections, we give the actual proof. 

\subsection{Symmetry reduction} \label{sec:sym-reduct} 
 
Let $X_n=\sH^{2}_{n}\times \HA{2}{\Div,0} $ and $Y_n=\Lpsi{2}{n}\times \LA{p}{\Div,0}$. 
Define the action of $H(\cL)$ on our spaces by
\begin{align}\label{rho-action-gen}
	\rho_g  (\psi(x), \al(x)) = (\psi(g x), g^{-1}\al(g x)),
\end{align}
where $g \in H(\cL)$. (The groups we deal with are abelian, so \eqref{rho-action-gen} defines a representation.)
We begin with the following
\begin{lemma}\label{lem:Xn-invar}
 Let $n$ be even.  Then 
$\rho_{g}: X_n\ra X_n\ \forall g\in SH(\cL)$.\end{lemma}
\begin{proof} By the definition, it suffices to show that if $\psi$ satisfies \eqref{gaugeperiod-psi'}. 
We check this condition. Recalling \eqref{gaugeperiod-psi'}, we have $\psi(g(x + t)) = e^{i(\frac{ n}{2} g x\cdot J g t+c_{g t}) }\psi(g x)$. One can compute easily that $g x\cdot J g t= x\cdot J t, \forall g\in SH(\cL)$. Furthermore by \eqref{cs-express}, we have $e^{ic_{g t}}= e^{ic_t}$, for $n$ even, and hence the result follows. \end{proof}
Let $F(\lambda, u), u=(\psi, \al)$, be the map defined in \eqref{F},  and $u_0=(0, 0)$, the normal state. Since the map $F(\lam, u)$ is rotationally,  translationally and gauge covariant and $\rho_{g}a^n=a^n\ \forall g\in H(\cL)$ (recall, $ a = a^n + \alpha$), 
we have the following lemma 
\begin{lemma}\label{lem:F-covar}
 Let 
  $\tilde T^{\rm gauge}_\chi: (\psi, \al) \ra (e^{i \chi} \psi, \al)$. Then 
\begin{align} \label{cov-F-gen} &F(\lam, \rho_{g} u)=\rho_{g} F(\lam, u),\  g \in H(\cL),\\  &
 F(\lam, \tilde T^{\rm gauge}_\chi u)=\tilde T^{\rm gauge}_\chi F(\lam, u),\  \forall \chi \in \R. 
 \end{align}
\end{lemma}
Define $\tilde\rho_{g}:= \tilde T^{\rm gauge}_{-\chi_g}\rho_{g}$, where  $\chi_g$ are some constants. (If $\chi_g$ is a representation of $G$, then so is $\tilde\rho_{g}$. By Proposition \ref{prop:xig-xindep}, we do not loose any generality by assuming $\chi_g$ are constants.) 
The lemma above leads to the following 
\begin{proposition}\label{prop:NulldF-invar}   Let $n$ be even. Then for $\rho$ given in \eqref{rho-action-gen}, we have
\begin{align} \label{NuldF-invar} \tilde\rho_g \nul_{X_n} dF(\lam, u)=\nul_{X_n} dF(\lam, \tilde \rho_g u),\ \forall g\in H(\lat). 
\end{align}   
Hence, if  $G$ is a subgroup of $H(\lat)$ and $\tilde\rho_g u_* =u_*, \forall g\in G$, then the subspace $\nul dF(\lam, u_*)$ is invariant under  the action of $G$.    \end{proposition}
\begin{proof} 
 Indeed, differentiating $F(\lam, \tilde\rho_g u)= \tilde\rho_g F(\lam, u)$ w.r. to $u$, we obtain
  \[d F(\lam, \tilde\rho_g u) \rho_g \xi =\tilde\rho_g d F(\lam, u)\xi,\] which gives \eqref{NuldF-invar}.  
\end{proof}

  Clearly,  $\tilde\rho_{g} u_0=u_0,\ \forall g \in H(\lat)$, for the normal state $u_0:= (0, 0)$, so, by Proposition \ref{prop:NulldF-invar}, 
 \begin{align} \label{dF-invar-gen}\nul_{X_n} dF(\lambda, u_0)\ \text{ is invariant under }\ \rho_{g},\ \forall g \in H(\lat). \end{align}

Recall that $ dF(\lambda, u_0)= A_{ \lambda}$. By formula \eqref{NullA-NullL}, it suffices to concentrate on $\Null_{\sH^{2}_{n}} (L^n -  n)$.  The action $\rho_{g}$ induces the action, $\rho_{g}'$ on$\psi$'s: 
\begin{align} \label{rho'-action-gen}	\rho_{g}' \psi(x)=  \psi(g^{-1} x),\  \forall g\in SH(\cL). 
\end{align}
Since $\Null_{\sH^{2}_{n}\times \HA{2}{\Div,0}} A_{ n}$ is invariant under $\tilde \rho_{g}$ and due to formula \eqref{NullA-NullL}, we conclude 
\begin{corollary}\label{cor:NullL-invar}  Let $n$ be even. Then  $ \Null_{\sH^{2}_{n}} (L^n -  n)$ is invariant under 
 the gauge and \eqref{rho'-action-gen} transformations, and therefore under $\tilde\rho_{g}',\   \forall g\in SH(\cL)$, where $\tilde\rho_{g}':=  e^{- i \chi_g}\rho_{g}'$, the restriction of  $\tilde\rho_{g}$ to $\psi$'s. 
\end{corollary} 

For a subgroup $G\subset G(\cL)$, we require that a solution in question is \textit{$G-$equivariant} w.r.to this action, in the sense that it satisfies
\begin{align}\label{u-Gequiv}
	\rho_g u = \tilde T^{\rm gauge}_{\chi_g} u.
\end{align}
where $u= (\psi, \al)$ and $\tilde T^{\rm gauge}_\chi: (\psi, \al) \ra (e^{i \chi} \psi, \al)$, for some functions $\chi_g$ (satisfying the corresponding co-cycle condition).  (It turns out it is sufficient to assume that $\chi_g$ are constants, see Proposition \ref{prop:xig-xindep}.)

 Note that, 
if  $u= (\psi, \al)$ satisfies Eq \eqref{u-Gequiv}, then $\psi$ obeys the equivariance condition
\begin{align}\label{psi-Gequiv}
	\rho_g' \psi = \xi_g \psi,\  \quad 
	\xi_g:= e^{i \chi_g},\  \quad g\in G.
\end{align}

Now, let $G$ be a subgroup of $H(\lat)$ 
with the irreducible representations labeled 
 by $\s$. We define  the subspaces 
\begin{align}
& X_{n \s}\subset X_n: 
\tilde\rho\big|_{X_{n \s}}\ \text{ is multiple of }\ \tilde\rho^\s,\\ 
&Y_{n \s} \subset Y_n: 
 \tilde\rho\big|_{X_{n \s}}\ \text{ is multiple of }\ \tilde\rho^\s. 
\end{align}
Then  $F : \R\times X_{n \s} \rightarrow Y_{n \s}$. Now, our goal is to choose $G$ and $\s$ such that 
  $\nul_{X_{n \s}} dF(\lambda, u_0)$  is one-dimensionall at the bifurcation point $\lambda=n$. Then Theorem \ref{thm:bif-thm-Kdim1}, with  the spaces $X$ and $Y$, appearing in \eqref{XY-cond},  given by $X=X_{n \s}$ and $Y=Y_{n \s}$, would be applicable and would give the desired result,  Theorem \ref{thm:MultiFluxExist}.
 
 Note that for any $G$ with $\rho_g u_0=u_0, \forall g\in G$, the bifurcation equation \eqref{bif-eqn} is invariant under $\rho_g$,
   \begin{equation}  \label{bif-eqn-invar-rho}
        \gamma(\lambda, \rho_g v)= \gamma(\lambda, v).
    \end{equation}


\subsection{Discrete Subgroups of $SO(2)$}
As was discussed above the maximal symmetry group of $\Null_{X_n} A_{ n}$ is  the group $G(\cL)\cap SO(2)=H(\cL)\cap SO(2)$. 
The Crystallographic restriction theorem says that $H(\lat)$ is either the cyclic, $C_k$, or dihedral, $D_k$, group, with $k=1, 2, 3, 4, 6$. Above, we ruled out $D_k$.  The case $k=1$ is trivial and gives us nothing new. Hence as a  symmetry group, $G$, 
 we take one of  the cyclic group of rotations, $C_k$, of order $k= 2, 3, 4, 6$.

  For $k=3$, the lattice whose symmetry group is $C_3$ is the hexagonal lattice. So it is to our advantage to consider $C_6$ instead for a stronger symmetry reduction. The case $k=4$ corresponds to square lattice, the proof of existence in this case is similar to the case $k=6$ but requires a smaller selection of flux $n$'s.
Thus, we consider only $C_2$ and $C_6$.  
 
 The group $C_k$ is  generated by a rotation $R_{k} \in SO(2)$ by the angle $2\pi/k$. If we identify $\R^2$ with $\C$, under $(x_1,x_2) \leftrightarrow x_1+ix_2$, then $R_{k}$ is identified  with the multiplication by  \[\xi_k = e^{2\pi i/k} \in U(1).\]  
 
We can specify the action \eqref{rho-action-gen} and \eqref{rho'-action-gen} to the present group by defining 
\begin{align}\label{rho}	&\rho_{k}  (\psi(x), \alpha(x)) = (
\psi(R_{k}^{-1} x), R_{k} \alpha(R_{k}^{-1} x)), \\
& \rho_{k}'  \psi(x)= \psi(R_{k}^{-1} x), 
\end{align}
where $k \in \F{Z}$. Then the equivariance conditions \eqref{u-Gequiv} and \eqref{psi-Gequiv} become, respectively,
\begin{align}\label{u-Gequiv'}
	\rho_k u = \tilde T^{\rm gauge}_{r\chi_k} u,\  \xi_k:=e^{i \chi_k},\ 
 \quad 	\rho_k' \psi = \xi_k^r \psi. 
\end{align}
Thus the group representation problem is eventually reduced to the eigenvalue problem for the operator $\rho_k'$.

\subsection{Spaces $X$ and $Y$}

	Since the groups  $C_k$ are  finite abelian groups, their irreducible unitary representations are 1-dimensional and, on $\cH^2$, coincide with the eigenspaces of the operator $\rho_{k}'$. Since $\rho_{k}'$ is unitary and satisfies 
	\begin{align}\label{rho'-nilp}(\rho_{k}')^k=\one,\end{align} it has exactly $k$  eigenvalues, $\xi_k^r = e^{2\pi i r/k}, r=0, \dots k-1$.  In this case, we specify our spaces for $(\psi, \al)$'s as
\begin{align}
& X_{n, k, r}:=\{u\in X_n: \rho_{k} u = \tilde T^{\rm gauge}_{r\chi_k} u\},\\ 
&Y_{n, k, r}=\{u\in  Y_n: \rho_{k} u =  \tilde T^{\rm gauge}_{r\chi_k} u\}. 
\end{align}
and the corresponding  spaces for $\psi$'s as:
 \begin{align}
& X_{n, k,  r}':=\{\psi\in \sH^{2}_{n}: \rho_{k}' \psi = \xi_k^r \psi\},\\ 
&Y_{n, k, r}'=\{\psi\in  \Lpsi{2}{n}: \rho_{k}' \psi = \xi_k^r \psi\}. 
\end{align}
Then, by Lemma \ref{lem:F-covar},  $F : \R\times X_{n, k,r} \rightarrow Y_{n, k, r}$, so condition  \eqref{XY-cond} holds.

\subsection{Multiplicity (Spaces $V_{n, k, r}$)} 
Let $n$ be the flux quantum number. 
For $n= 1, 2, \dots, k= 1, 2, \dots, r=  0, 1, 2, \dots, k-1$, we define the spaces  
 \begin{align}\label{tildeV-V}\tilde{V}_n := \Null_{X_n} (L^n -  n)\ \text{ and }\ \tilde{V}_{n, k, r} := \Null_{X_{n, k,  r}'} (L^n - n).\end{align} 
Our first goal is to prove the following 
\begin{theorem} \label{nEvenTildeV-classif}
 Let $n$ be even. Then $\tilde V_{n,k,r}$ is one dimensional for $k=6$ and for  the pairs 
\begin{align}
	(n, r) =& (2, 0), (2, 2), (4, 0), (4, 1), (4, 2), (4,4), \\
		& (6,1), (6,2), (6,3), (6,4) \\
		& (8,1), (8,3), (8,4), (8,5) \\
		& (10, 3), (10,5)
\end{align}
\end{theorem}
 To prove this theorem, we  pass to the corresponding spaces of theta functions. The latter are more rigid since they are holomorphic.  

By the definition, the space $\tilde{V}_n$ is related to  the space $V_n$, 
 defined in Proposition \ref{prop:Landau-ham-spec}  
  as \begin{align} \label{Vn-tildeVn} \tilde{V}_n =f_n V_n,\  
  f_n(x):=e^{\frac{in}{2}x_2(x_1 + ix_2) }=e^{-\frac{C}{2}(|z|^2-z^2)}, \end{align}
where $C:= \frac{\pi n}{\im \tau}$ and $z:=  (x_1+i x_2)/ \sqrt{\frac{2\pi}{\im\tau} }$, or in terms of the functions,
\begin{align}\label{transfer}	\psi(x) = f_n (z)\theta(z),\  f_n (z):=e^{\frac{in}{2}x_2(x_1 + ix_2) }=e^{-\frac{C}{2}(|z|^2-z^2)}.   \end{align}
 Elements, $\theta$, of the subspace $V_n$, will be called \textit{$n$-theta functions}.  Similarly, we define the spaces $V_{n, k, r}$ by
\begin{align} \label{Vnkr-tildeVnkr} 
\tilde{V}_{n, k, r}=f_n V_{n, k, r}. 
 \end{align}
We define the induced action on theta functions via $T_{n, k}:= f_n ^{-1}\tilde\rho_{k,j}' f_n $. We have 
\begin{lemma} \label{Tevs} 
The operator $T_{n, k}$ 
is unitary and satisfies $(T_{n, k})^k=\one$. Consequently, its spectrum consists of the eigenvalues of the form $\xi_k^r$ for some $r =0, \dots, k-1$. Moreover, the eigenfunctions corresponding to the eigenvalue $\xi_k^r$ has zero at $z=0$ of the order $r$.
\end{lemma}
 \begin{proof}  
Eq. \eqref{rho'-nilp} and the definition $T_{n, k}:= f_n ^{-1}\tilde\rho_{k,j}' f_n $ 
To show the second claim, let $\lambda$ be any eigenvalue. Expanding $\theta(z) = az^m + O(|z|^{k+1})$, where $a \not= 0$ and $m \geq 0$, and $e^x = 1 + O(|x|)$ and writing out the eigenvalue equation, we see that to lowest order in $z$,
\begin{align}
	\lambda a z^m =  a \xi^m_k z^m
\end{align}
Hence  $\lambda = \xi^m_k$.
\end{proof}
 \begin{corollary}\label{lem:VeigenspT}
 Let $n$ be even. Then  $V_{n, k, r}$ are eigenspaces of the operator $T_{n, k }$ corresponding to the eigenvalues $\xi_k^r$. \end{corollary}

We recall that the Wigner--Seitz cell around a lattice point is defined as the locus of points in space that are closer to that lattice point than to any of the other lattice points. To eliminate the overlap between the Wigner--Seitz cells around  different points, we agree on the choice of their boundaries. Say, observing that the Wigner--Seitz cell is a {\it(slanted)} hexagon, 
we set the boundary of a  Wigner--Seitz cell to contain the {\it three left-most edges and the two left-most vertices} (see e.g. Fig \ref{fig:Theta0}). 
Hence  Wigner--Seitz cells 
 of a lattice tile $\R^2$ without an intersection.

By a standard result about theta functions (see Theorem \ref{unique} of Appendix \ref{sec:theta-fns}) or line bundles,  theta functions are entirely determined by their zeros, $z_j$, and multiplicities, $m(z_j)$, in a  Wigner--Seitz cell, $W$. By analogy with holomorphic sections of line bundles, we call the collection of  zeros and multiplicities of a  theta function, $\theta$, its divisor and denote div$(\theta)=\sum_{z\in W} m(z) z$. The degree of a  theta function, $\theta$, is defined as the degree of its divisor,  $|div(\theta)|=\sum_{z\in W} m(z) $. Then $\theta\in V_n \iff |div(\theta)|=n$.

Corollary \ref{lem:VeigenspT} 
 and Lemma \ref{Tevs}  and standard results about theta functions mentioned above imply 
\begin{corollary}\label{cor:theta-zeros}
 $\theta\in V_{n, k, r} \iff$ the following three conditions hold: (a) $|div(\theta)|=n$ (i.e. $\theta$ has $n$ zeros counting their the multiplicities); (b) $m(0)=r$ (i.e. $\theta$ has the zero of the multiplicity $r$ at the origin); (c) $div(T_{n, k}(\theta))=div(\theta)$ (i.e. $div(\theta)$ is invariant under the transformation $T_{n, k}$ (i.e. rotation by $2\pi/k$)). \end{corollary}

\subsubsection{$C_6$}

By Corollary \ref{cor:theta-zeros}, we want to translate the eigenvalue problem $T_{n,6} \theta = \xi^r \theta$ into the existence of divisors corresponding to the zeros of $\theta$. This would allows us to find 1-1 correspondence between all such $\theta$ and simple diagrams for our analysis.

Let div (divisor) denote a finite collection of points in the Wigner-Seitz cell $W$, centered at the origin, together with their multiplicities, i.e. a map from $W$ to $\Z^+$ with a finite number of non-zero values. We can identify the divisors with the diagrams as in Fig \ref{fig:diagrams} (the WS cell with the choice of points and multiplicities), the latter provide handy illustrations. Then we obtain a map
\begin{align} \label{eqn:Divv}
	\Divv : \text{theta functions} \ra \text{divisors/diagrams},
\end{align}
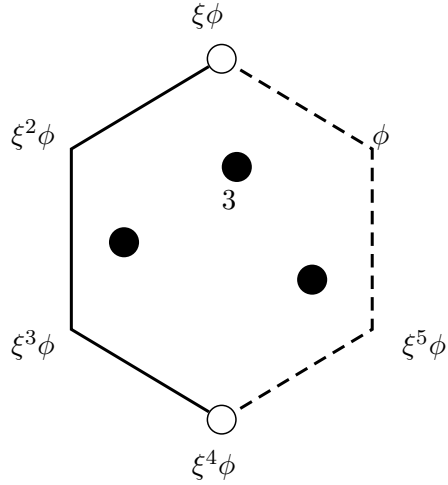
\begin{figure}[h]
	\centering
		\psscalebox{1.0 1.0} 
{
\begin{pspicture}(0,-3.16)(6.85,3.16)
\psline[linecolor=black, linewidth=0.04](2.8,2.44)(0.8,1.24)(0.8,-1.16)(2.8,-2.36)(2.8,-2.36)
\psline[linecolor=black, linewidth=0.04, linestyle=dashed, dash=0.17638889cm 0.10583334cm](2.8,2.44)(4.8,1.24)(4.8,-1.16)(2.8,-2.36)
\psdots[linecolor=black, dotstyle=o, dotsize=0.4, fillcolor=white](2.8,2.44)
\psdots[linecolor=black, dotstyle=o, dotsize=0.4, fillcolor=white](2.8,-2.36)
\rput[bl](4.8,1.24){$\phi$}
\rput[bl](2.4,2.84){$\xi \phi$}
\rput[bl](0.0,1.24){$\xi^2 \phi$}
\rput[bl](0.0,-1.56){$\xi^3 \phi$}
\rput[bl](2.4,-3.16){$\xi^4 \phi$}
\rput[bl](5.2,-1.56){$\xi^5 \phi$}
\psdots[linecolor=black, dotsize=0.4](3,1)
\rput[bl](2.8,0.44){$3$}
\psdots[linecolor=black, dotsize=0.4](1.5,0)
\psdots[linecolor=black, dotsize=0.4](4,-0.5)
\end{pspicture}}
\caption{Typical diagram of a divisor. The black dots denote nonzero point on $W$. Each black dot is assumed to have multiplicity $1$ unless otherwise indicated by a number next to it.}
\label{fig:diagrams}
\end{figure}
Since we are interested in eigenvectors of $T_{n,6}$, we restrict $\divv$ to the set of eigenvalues of $T_{n,6}$. In particular, let \[V_{n, k, r}^{\divv} :=\{\divv: |\divv|=n,\  |\divv(0)| \equiv r \bmod k,\  T_{n, k}\divv = \divv\}.\] We have the following result, proven in Section \ref{sec:V-Vdiv}:

\begin{theorem}[Classification Theorem for $C_6$-invariant Theta Functions] \label{thm:Gclassify-theta}
The map $\Divv : V_{n, 6, r} \ra V_{n, 6, r}^{\divv}$ is a bijection, and in particular \[\dim V_{n, 6, r}=\dim V_{n, 6, r}^{\divv}.\] 
\end{theorem}

To compute $\dim V_{n, k, r}^{\divv}$ it is convenient to give each point of $W$ the index which is the number of elements in the orbit under $T_{n, k}$ generated by this point. Thus, for $k=6$, all interior points of $W$ and all boundary points, besides the vertices and the midpoints of the edges, have the index $6$.  The boundary vertices and the midpoints of the edges  have the indices $2$ and $3$, respectively, and the origin has the index $1$.

By the orbit-stabilizer theorem, there is no divisor with index 4 or 5 where the multiplicity is simple at each point, since $4$ and $5$ do not divide $6$.

We identify orbits with the same index. Denote the multiplicity of points in the orbit of the index $i$ by $m_i$, so that $m_1=r$. Then we have the relation
\begin{align}\label{mi-cond}
	\sum_i i m_i \equiv 1 \cdot m_1+2 m_2+2 m_3+6 m_6=n.
\end{align}
We use this equation to classify the diagrams to obtain 

\begin{theorem} \label{nEvenV-classif}
 Let $n$ be even. Then 	$V_{n,k,r}$ is one dimensional for $k=6$ and for 
	 the pairs 
	\begin{align}
	(n, r) =& (2, 0), (2, 2), (4, 0), (4, 1), (4, 2), (4,4), \\
		& (6,1), (6,2), (6,3), (6,4) \\
		& (8,1), (8,3), (8,4), (8,5) \\
		& (10, 3), (10,5)
\end{align}  
\end{theorem}	
This result implies Theorem \ref{nEvenTildeV-classif}.
A table describing the explicit spanning theta functions for $V_{n,6,r}$ can be found in Appendix \ref{app:ThetaTable}

\subsubsection{$C_2$}
For $\xi_2 = -1$, the corresponding irreducible representations of $C_2$ are simply even and odd functions. By the correspondence \eqref{transfer},
the evenness and oddness of $\psi$ translates to the same property of $\theta$. Hence, we easily see that linear compatibility is satisfied as $V_n$ can be decomposed into odd and even functions 

\begin{lemma} 
Let $V_n=V_{n, \rm even}\oplus V_{n, \rm odd}$ be the decomposition of  $V_n$ into even and odd functions. Then $\dim V_{n, \rm even/\rm odd}\ge 1$
\end{lemma}
\begin{proof}
Let $\theta_0,...,\theta_{n-1}$ be the standard basis for the set of theta functions as in Theorem \ref{Basis}. We recall that
\begin{align}
	\theta_m(-z) = \theta_{n-m \bmod n}(z)
\end{align}
This shows that the set of odd theta functions are spanned by
\begin{align}
	\sigma_{j,-}(z) = \theta_j(z) - \theta_{n-j}(z)
\end{align}
Similarly, the even functions are spanned by
\begin{align}
	\sigma_{j,+}(z) = \theta_j(z) + \theta_{n-j}(z)
\end{align}
\end{proof} 

\begin{corollary}
If $n=3$, then $\dim V_{n,k,1} = \dim V_{n,odd} = 1$.
\end{corollary}

\subsection{Irreducibility}
Using \eqref{Vn-tildeVn}, irreducibility of $\psi$ translate to irreducibility of $\theta$. We say that $\theta$ is \textit{reducible (to $\cL'$)} if there is a finer lattice, $\cL'$, containing $\cL$ s.t. the corresponding $\psi$ is gauge periodic with respect to $\cL'$. Otherwise we say that $\theta$ is \textit{irreducible}. We now proof irreducibility below.

\subsubsection{Irreducibility for $C_6$ Symmetry}

\begin{theorem} \label{nEvenLin}
The spanning theta functinon of $V_{n,k,j}$ is  is irreducible 
	 for  the pairs
\[
	(n, j)=(4, 0), (6, 3), (8, 5), (10,5),
\]
\end{theorem}
	
\begin{proof}
To prove irreducibility, we need the following basic lemma:
\begin{lemma} \label{lem:FinerLattice}
Let $\mathcal{L} \subset \mathcal{L}'$ be lattices. Let $\Omega_{\mathcal{L}}$ be any fundamental cell of $\mathcal{L}$. Then precisely one of the following holds: there is a $v \in \mathcal{L}'$ such that $v \in \Omega_{\mathcal{L}} \backslash \mathcal{L}$ or $\mathcal{L}' =\mathcal{L}$.
\end{lemma}
\begin{proof}
Assume that no such $v \in \mathcal{L}'$ with $v \in \Omega_{\mathcal{L}} \backslash \mathcal{L}$ exists. That is, every $v \in \mathcal{L}'$ such that $v \in \Omega_{\mathcal{L}}$ is contained in $\mathcal{L}$. Since translates of $\Omega_{\mathcal{L}}$ tiles the entire plane and $\mathcal{L} \subset \mathcal{L}'$, we conclude that every element of $\mathcal{L}'$ is in $\mathcal{L}$. That is, $\mathcal{L} = \mathcal{L}'$.
\end{proof}

Now, 
 by choice of theta functions indicated in table \eqref{theta-table}, we see that for vortex number $n$, the number of zeros of the chosen theta at the origin differs from the number of zeros at any other point in $\Omega_{\mathcal{L}}$. If $\mathcal{L}'$ is any finer lattice containing $\mathcal{L}$ with respect to which our solution is gauge periodic, then Lemma \ref{lem:FinerLattice} implies that the number of flux per fundamental cell of $\mathcal{L}$ for our chosen theta $= (\text{number of zero at the origin}) + n > n$. This is a contradiction.

\end{proof}

\subsubsection{Irreducibility of odd theta functions with prime flux} 

\begin{proposition}
Let $\theta$ be an odd theta function with prime flux $p$. Then $\theta$ is irreducible.
\end{proposition}
\begin{proof}
Let $\theta$ be gauge periodic with respect to $\mathcal{L}$. Let $\mathcal{L} \subset \mathcal{L}'$ be any finer lattice. Let $q$ denote the number of zeros of $\theta$ in a fundamental cell of $\mathcal{L}'$. We first claim that $q \mid p$. Let $u,v$ be the generators of $\mathcal{L}'$ and $\Omega_{\mathcal{L}'}$ be the fundamental cell of $\mathcal{L}'$ formed by taking the convex hall of $u$ and $v$ (together with appropriate boundary). Define an equivalence relationship as follows: two translates of $\Omega_{\mathcal{L}'}$, $s + \Omega_{\mathcal{L}'}$ and $s' + \Omega_{\mathcal{L}'}$ for $s,s' \in \mathcal{L}$, are said to be equivalent if $s-s' \in \mathcal{L}$. Let $s_1 + \Omega_{\mathcal{L}'},...,s_k + \Omega_{\mathcal{L}'}$ be maximally inequivalent for $s_1,...,s_k \in \mathcal{L}'$. Since translates of $\Omega_{\mathcal{L}'}$ tile the entire plane, we conclude that by appropriate translates of the $s_j$'s, $s_1 + \Omega_{\mathcal{L}'} \cup ... \cup s_k + \Omega_{\mathcal{L}'}$ is a fundamental domain of $\mathcal{L}$. In particular, $p = qk$ since $\theta$ has the same number of zeros in each fundamental cell of $\mathcal{L}'$. Since $p$ is prime, either $q=p$ or $q=1$. If $q=p$, then $\mathcal{L} = \mathcal{L}'$. Otherwise $q=1$. Since $0, 1/2, \tau/2$ are zeros of $\theta$ and each fundamental cell $s+ \Omega_{\mathcal{L}'}$ has exactly one zero, we conclude that $0,1/2, \tau/2 \in \mathcal{L}'$. So in particular $\mathcal{L} \subset \frac{1}{2}(\Z + \tau \Z)$ and $\psi$ is gauge periodic with respect to $\frac{1}{2}(\Z + \tau \Z)$. This is clearly not possible, for otherwise $4 \mid p$ is a contradiction.
\end{proof}

\section{Proof of Theorem \ref{thm:Gclassify-theta}}\label{sec:V-Vdiv}
\begin{figure}[h] 
\hspace{1cm}

\psscalebox{1.0 1.0} 
{
\begin{pspicture}(0,-2.601389)(2.601389,2.601389)
\psline[linecolor=black, linewidth=0.04](2.4,2.4)(0.4,1.2000002)(0.4,-1.1999998)(2.4,-2.3999999)(2.4,-2.3999999)
\psdots[linecolor=black, dotstyle=o, dotsize=0.4, fillcolor=white](2.4,2.4)
\psdots[linecolor=black, dotstyle=o, dotsize=0.4, fillcolor=white](2.4,-2.3999999)
\rput[bl](0.0,1.6000001){$\xi^2 \phi$}
\rput[bl](0.0,-1.9999999){$\xi^3 \phi$}
\psdots[linecolor=black, dotsize=0.4](0.4,0.40000013)
\psdots[linecolor=black, dotsize=0.4](0.4,-0.39999986)
\psellipse[linecolor=black, linewidth=0.04, shadow=true,shadowsize=0.0058, fillcolor=black, dimen=outer](0.95,1.55)(0.2,0.2)
\psellipse[linecolor=black, linewidth=0.04, shadow=true,shadowsize=0.0058, fillcolor=black, dimen=outer](1.7,1.95)(0.2,0.2)
\psellipse[linecolor=black, linewidth=0.04, shadow=true,shadowsize=0.0058, fillcolor=black, dimen=outer](0.9,-1.5)(0.2,0.2)
\psellipse[linecolor=black, linewidth=0.04, shadow=true,shadowsize=0.0058, fillcolor=black, dimen=outer](1.65,-1.9)(0.2,0.2)
\psline[linecolor=black, linewidth=0.04, linestyle=dashed, dash=0.17638889cm 0.10583334cm](2.6,2.44)(4.6,1.24)(4.6,-1.16)(2.6,-2.36)
\end{pspicture}

\hspace{3cm}

\begin{pspicture}(0,-2.601389)(2.601389,2.601389)
\psline[linecolor=black, linewidth=0.04](2.8,2.44)(0.8,1.24)(0.8,-1.16)(2.8,-2.36)(2.8,-2.36)
\psline[linecolor=black, linewidth=0.04, linestyle=dashed, dash=0.17638889cm 0.10583334cm](2.8,2.44)(4.8,1.24)(4.8,-1.16)(2.8,-2.36)
\psdots[linecolor=black, dotstyle=o, dotsize=0.4, fillcolor=white](2.8,2.44)
\psdots[linecolor=black, dotstyle=o, dotsize=0.4, fillcolor=white](2.8,-2.36)
\rput[bl](4.8,1.24){$\phi$}
\rput[bl](2.4,2.84){$\xi \phi$}
\rput[bl](0.0,1.24){$\xi^2 \phi$}
\rput[bl](0.0,-1.56){$\xi^3 \phi$}
\rput[bl](2.4,-3.16){$\xi^4 \phi$}
\rput[bl](5.2,-1.56){$\xi^5 \phi$}
\rput[bl](2.7,0){$6$}
\end{pspicture}
}
\caption{Index 6 divisors. The figure on the left has six distinct dots on its left most 3 edges, forming an orbit for $C_6$. The figure on the r.h.s. indicates six distinct dots forming an orbit of $C_6$ in the interior of the WS cell.}
\label{fig:WSIndex6} 
\end{figure}
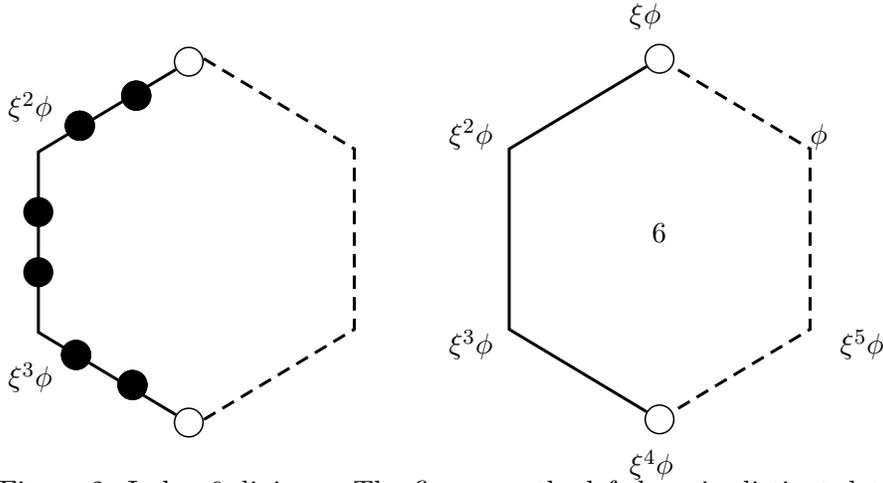
\begin{figure}[h]
	\centering
		\psscalebox{1.0 1.0} 
{
\begin{pspicture}(0,-3.16)(6.85,3.16)
\psline[linecolor=black, linewidth=0.04](2.8,2.44)(0.8,1.24)(0.8,-1.16)(2.8,-2.36)(2.8,-2.36)
\psline[linecolor=black, linewidth=0.04, linestyle=dashed, dash=0.17638889cm 0.10583334cm](2.8,2.44)(4.8,1.24)(4.8,-1.16)(2.8,-2.36)
\psdots[linecolor=black, dotstyle=o, dotsize=0.4, fillcolor=white](2.8,2.44)
\psdots[linecolor=black, dotstyle=o, dotsize=0.4, fillcolor=white](2.8,-2.36)
\rput[bl](4.8,1.24){$\phi$}
\rput[bl](2.4,2.84){$\xi \phi$}
\rput[bl](0.0,1.24){$\xi^2 \phi$}
\rput[bl](0.0,-1.56){$\xi^3 \phi$}
\rput[bl](2.4,-3.16){$\xi^4 \phi$}
\rput[bl](5.2,-1.56){$\xi^5 \phi$}
\psdots[linecolor=black, dotsize=0.4](2.8,0.04)
\rput[bl](2.8,0.44){$2$}
\end{pspicture}
}
\caption{Pictorial discription of $\theta_{2}$}
\label{fig:Theta2}
\end{figure}
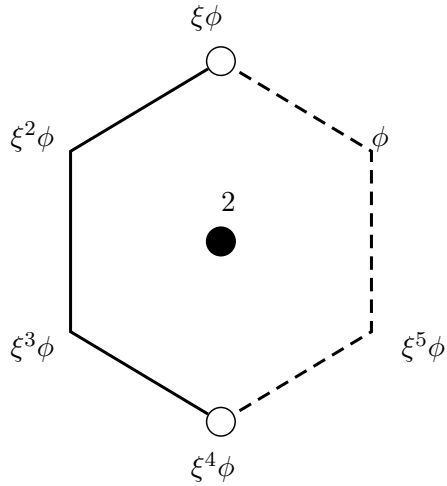
\begin{figure}[h]
	\centering
		\psscalebox{1.0 1.0}
{
\begin{pspicture}(0,-3.16)(6.85,3.16)
\psline[linecolor=black, linewidth=0.04](2.8,2.44)(0.8,1.24)(0.8,-1.16)(2.8,-2.36)(2.8,-2.36)
\psline[linecolor=black, linewidth=0.04, linestyle=dashed, dash=0.17638889cm 0.10583334cm](2.8,2.44)(4.8,1.24)(4.8,-1.16)(2.8,-2.36)
\psdots[linecolor=black, dotstyle=o, dotsize=0.4, fillcolor=white](2.8,2.44)
\psdots[linecolor=black, dotstyle=o, dotsize=0.4, fillcolor=white](2.8,-2.36)
\rput[bl](4.8,1.24){$\phi$}
\rput[bl](2.4,2.84){$\xi \phi$}
\rput[bl](0.0,1.24){$\xi^2 \phi$}
\rput[bl](0.0,-1.56){$\xi^3 \phi$}
\rput[bl](2.4,-3.16){$\xi^4 \phi$}
\rput[bl](5.2,-1.56){$\xi^5 \phi$}
\psdots[linecolor=black, dotsize=0.4](0.8,1.2000002)
\psdots[linecolor=black, dotsize=0.4](0.8,-1.1999998)
\end{pspicture}
}
\caption{Pictorial description of $\theta_0$}
\label{fig:Theta0}
\end{figure}
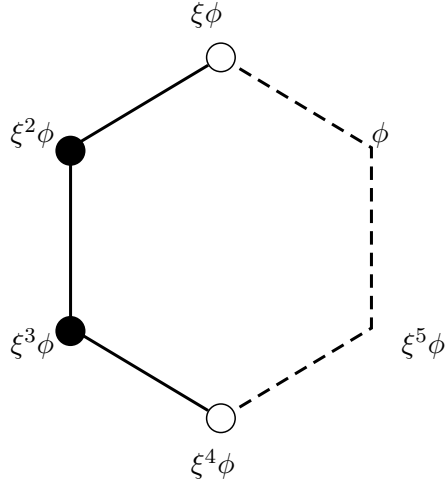

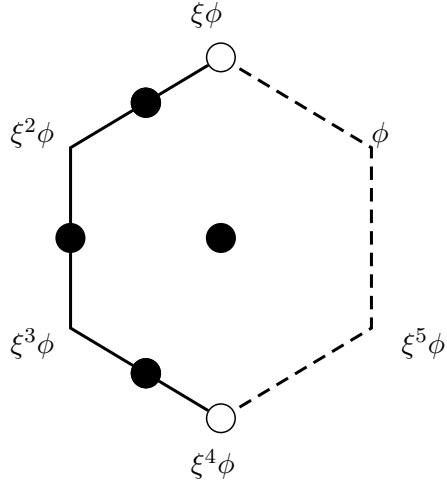
\begin{figure}[h]
\centering
\psscalebox{1.0 1.0} 
{
\begin{pspicture}(0,-3.16)(6.85,3.16)
\psline[linecolor=black, linewidth=0.04](2.8,2.44)(0.8,1.24)(0.8,-1.16)(2.8,-2.36)(2.8,-2.36)
\psline[linecolor=black, linewidth=0.04, linestyle=dashed, dash=0.17638889cm 0.10583334cm](2.8,2.44)(4.8,1.24)(4.8,-1.16)(2.8,-2.36)
\psdots[linecolor=black, dotstyle=o, dotsize=0.4, fillcolor=white](2.8,2.44)
\psdots[linecolor=black, dotstyle=o, dotsize=0.4, fillcolor=white](2.8,-2.36)
\rput[bl](4.8,1.24){$\phi$}
\rput[bl](2.4,2.84){$\xi \phi$}
\rput[bl](0.0,1.24){$\xi^2 \phi$}
\rput[bl](0.0,-1.56){$\xi^3 \phi$}
\rput[bl](2.4,-3.16){$\xi^4 \phi$}
\rput[bl](5.2,-1.56){$\xi^5 \phi$}
\psellipse[linecolor=black, linewidth=0.04, shadow=true,shadowsize=0.0058, fillcolor=black, dimen=outer](1.8,1.84)(0.2,0.2)
\psellipse[linecolor=black, linewidth=0.04, shadow=true,shadowsize=0.0058, fillcolor=black, dimen=outer](1.8,-1.76)(0.2,0.2)
\psdots[linecolor=black, dotsize=0.4](2.8,0.04)
\psdots[linecolor=black, dotsize=0.4](0.8,0.04)
\end{pspicture}
}
\caption{Pictorial description of $\theta_1$}
\label{fig:Theta4}
\end{figure}
We note that $V^{\divv}$ is the subset of the $\Z$-module of divisors generated by diagrams of the form in figures \ref{fig:WSIndex6} - \ref{fig:Theta4} such that the multiplicity of each point is non-negative. So we show that there is a theta function corresponding to each such diagram. This would show that $V^{\divv} \subset \Divv(V^{EV})$. In what follows $\tau =\xi = e^{\pi i/3}$ as before. Existence of $\theta_2$ is a direct result of Theorem \ref{singular} with $n=2$. Namely, the set of permissible double zeros for theta functions in $V_2$ is just $\frac{1}{2}(\F{Z} +\tau \F{Z})$. 

To construct $\theta_0$, let $\theta_0,\theta_1$ be a basis for $V_2$ as in Theorem \ref{Basis}. Form the function
\begin{align}
	\sigma(z) := \det \Sigma(z) := \det \left( \begin{array}{cc}
																								\theta_{2,0}(z) & \theta_{2,1}(z) \\
																								\theta_{2,0}(-z) & \theta_{2,1}(-z) 
																							\end{array} \right)
\end{align}
By Theorem \ref{Basis}, $\theta_{2,i}(z)$ is symmetric about $0$ for $i=0,1$ when $n=2$, thus $\sigma(z) = 0$ identically. In particular, for $z_0 = \frac{1}{4}(\tau +1)$, there are constants $c_0,c_1$ such that $c_0\theta_{2,0} + c_1\theta_{2,1}$ has two simple zeros located at $z_0,-z_0$, respectively. This proves item 2. 

The theta function $\theta_4$ is the Wronskian, $\Theta$, of $\theta_0$ and $\theta_1$. For $n=2$, Theorem \ref{Wronskian} shows that the location of the zeros of $\Theta$ are precisely the set of permissible zeros for a singular $2$-theta function. In this case, it is $\frac{1}{2}(1+\tau)$ by Theorem \ref{singular}. It matches the definition of $\theta_4$.

Finally, we show existence of theta functions with $6$ distinct zeros on the WS-cell. Let $a_1,a_2,a_3 \in \F{C}$. Consider
\begin{align}
	\sigma(z) := \det \left( \begin{array}{cccc}
															\theta_{6,0}(z+a_1) & \theta_{6,1}(z+a_1) & \cdots 	& \theta_{6,5}(z+a_1) \\
															\theta_{6,0}(z-a_1) & \theta_{6,1}(z-a_1) & \cdots  & \theta_{6,5}(z-a_1) \\
															\theta_{6,0}(z+a_2) & \theta_{6,1}(z+a_2) & \cdots  & \theta_{6,5}(z+a_2) \\
															\theta_{6,0}(z-a_2) & \theta_{6,1}(z-a_2) & \cdots  & \theta_{6,5}(z-a_2) \\
															\theta_{6,0}(z+a_3) & \theta_{6,1}(z+a_3) & \cdots  & \theta_{6,5}(z+a_3) \\
															\theta_{6,0}(z-a_3) & \theta_{6,1}(z-a_3) & \cdots  & \theta_{6,5}(z-a_3) 
														\end{array} \right)
\end{align}
Recalling that $\theta_{n,m}(-z) = \theta_{n,n-m \bmod n}(z)$ by Theorem \ref{Basis}, we see that 
\begin{align}
	&\sigma(-z)\\ \notag	&:= \det \left( \begin{array}{cccccc}
															\theta_{6,0}(z-a_1) & \theta_{6,5}(z-a_1) & \theta_{6,4}(z-a_1) & \theta_{6,3}(z-a_1) & \theta_{6,2}(z-a_1) 	& \theta_{6,1}(z-a_1) \\
															\theta_{6,0}(z+a_1) & \theta_{6,5}(z+a_1)  && \cdots  && \theta_{6,1}(z+a_1) \\
															\theta_{6,0}(z-a_2) & \theta_{6,5}(z-a_2) && \cdots  && \theta_{6,1}(z-a_2) \\
															\theta_{6,0}(z+a_2) & \theta_{6,5}(z+a_2) && \cdots  && \theta_{6,1}(z+a_2) \\
															\theta_{6,0}(z-a_3) & \theta_{6,5}(z-a_3) && \cdots  && \theta_{6,1}(z-a_3) \\
															\theta_{6,0}(z+a_3) & \theta_{6,5}(z+a_3) && \cdots  && \theta_{6,1}(z+a_3) 
														\end{array} \right) \\
\notag 							&= (-1)^3 \det \left( \begin{array}{cccc}
															\theta_{6,0}(z+a_1) & \theta_{6,5}(z+a_1) & \cdots 	& \theta_{6,1}(z+a_1) \\
															\theta_{6,0}(z-a_1) & \theta_{6,5}(z-a_1) & \cdots  & \theta_{6,1}(z-a_1) \\
															\theta_{6,0}(z+a_2) & \theta_{6,5}(z+a_2) & \cdots  & \theta_{6,1}(z+a_2) \\
															\theta_{6,0}(z-a_2) & \theta_{6,5}(z-a_2) & \cdots  & \theta_{6,1}(z-a_2) \\
															\theta_{6,0}(z+a_3) & \theta_{6,5}(z+a_3) & \cdots  & \theta_{6,1}(z+a_3) \\
															\theta_{6,0}(z-a_3) & \theta_{6,5}(z-a_3) & \cdots  & \theta_{6,1}(z-a_3) 
														\end{array} \right) \\
							&= (-1)^3 (-1)^2 \sigma(z) \\
							&= -\sigma(z)
\end{align}
where the factor $(-1)^3$ arises from interchanging the $2i-1$ and $2i$-th row for $i=1,2,3$. The $(-1)^2$ factor occurs after interchanging the second and the 6-th column and interchanging the third and the fourth column. So we have that $\sigma(0) = 0$. This proves the desired claim that $\sigma(z)$ has a kernel.

To prove that $\Divv(V^{EV}) = V^{\divv}$, we study orbits of $C_6$ on the WS-cell. We will use the divisor and theta function picture (see (\ref{eqn:Divv})) interchangeably. We find all possible orbits of the action of $C_6$ on the WS-cell. The only choice of having index 1 is the case where the origin has index 1. The only index 2 possibility where each point has index one is shown in Fig. \ref{fig:Theta0}. Then we have index 3 divisors. The possible location of points on $W$ with index 3 each with multiplicity 1 is as shown in Fig. \ref{fig:WSInd3}

\begin{figure}[h]
\centering
\psscalebox{1.0 1.0} 
{

\begin{pspicture}(0,-3.16)(6.85,3.16)
\psline[linecolor=black, linewidth=0.04](2.8,2.44)(0.8,1.24)(0.8,-1.16)(2.8,-2.36)(2.8,-2.36)
\psline[linecolor=black, linewidth=0.04, linestyle=dashed, dash=0.17638889cm 0.10583334cm](2.8,2.44)(4.8,1.24)(4.8,-1.16)(2.8,-2.36)
\psdots[linecolor=black, dotstyle=o, dotsize=0.4, fillcolor=white](2.8,2.44)
\psdots[linecolor=black, dotstyle=o, dotsize=0.4, fillcolor=white](2.8,-2.36)
\rput[bl](4.8,1.24){$\phi$}
\rput[bl](2.4,2.84){$\xi \phi$}
\rput[bl](0.0,1.24){$\xi^2 \phi$}
\rput[bl](0.0,-1.56){$\xi^3 \phi$}
\rput[bl](2.4,-3.16){$\xi^4 \phi$}
\rput[bl](5.2,-1.56){$\xi^5 \phi$}
\psdots[linecolor=black, dotsize=0.4](0.8,0.0)
\psdots[linecolor=black, dotsize=0.4](1.8,1.8000002)
\psdots[linecolor=black, dotsize=0.4](1.8,-1.7999998)
\end{pspicture}
}
\caption{Index 3 divisor}
\label{fig:WSInd3}
\end{figure}
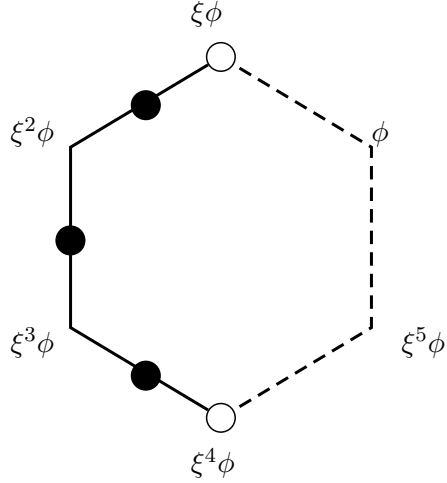
By the orbit-stablizer theorem, there is no divisor with index 4 or 5 where the multiplicity is simple at each point, since $4$ and $5$ do not divide $6$. Finally, we consider the index 6 case. The possible divisors are shown in Fig. \ref{fig:WSIndex6}.

Now we are ready for the proof of Theorem \ref{thm:Gclassify-theta}. First note that injectivity of the map $\Divv$ is a direct consequence of Proposition \ref{unique}. Now if $\theta$ is any $C_6$-equivariant theta function, its zeros are unions of orbits of $C_6$. 
We may divide $\theta$ by $C_6$-equivariant theta functions corresponding to elements of $V^{\divv}$ to produces new theta functions with fewer zeros in the Wigner-Seitz cell. Note that this division process preserves $C_6$-equivariance. We repeat this process until any further division results in a non-theta-function. We claim that the resulting function, $\sigma$, is a complex number.
If so, we have completely factor $\theta$ by theta functions from $V^{\divv}$ and the bijection is established.

Now, we study $\sigma$. $\sigma$ cannot have any zeros that form an orbit of $C_6$ of size 6, otherwise they can be removed by dividing by an element from $V^{\divv}$, contradicting the definition of $\sigma$. It can neither have zeros that form orbits of size 2 for the same reason. Hence, the zeros of $\sigma$ can only be in the following two configuration: one zero at the origin, or as shown in Fig. \ref{fig:WSInd3}. To see this, if there are zeros as in Fig. \ref{fig:WSInd3}, but with higher multiplicity, we divide $\sigma$ by $\theta_1^{-1} \theta_4^{2}$ as shown in Fig. \ref{fig:14Inv} to remove all the multiplicities. Likewise we can divide by $\theta_2$ to remove even multiplicity at the origin. 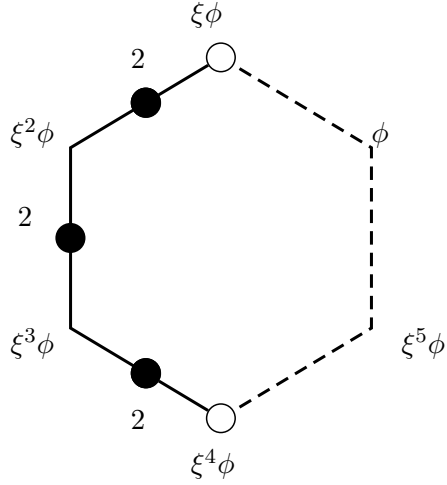
\begin{figure}[h]
\centering
\psscalebox{1.0 1.0} 
{
\begin{pspicture}(0,-3.16)(6.85,3.16)
\psline[linecolor=black, linewidth=0.04](2.8,2.44)(0.8,1.24)(0.8,-1.16)(2.8,-2.36)(2.8,-2.36)
\psline[linecolor=black, linewidth=0.04, linestyle=dashed, dash=0.17638889cm 0.10583334cm](2.8,2.44)(4.8,1.24)(4.8,-1.16)(2.8,-2.36)
\psdots[linecolor=black, dotstyle=o, dotsize=0.4, fillcolor=white](2.8,2.44)
\psdots[linecolor=black, dotstyle=o, dotsize=0.4, fillcolor=white](2.8,-2.36)
\rput[bl](4.8,1.24){$\phi$}
\rput[bl](2.4,2.84){$\xi \phi$}
\rput[bl](0.0,1.24){$\xi^2 \phi$}
\rput[bl](0.0,-1.56){$\xi^3 \phi$}
\rput[bl](2.4,-3.16){$\xi^4 \phi$}
\rput[bl](5.2,-1.56){$\xi^5 \phi$}
\psellipse[linecolor=black, linewidth=0.04, shadow=true,shadowsize=0.0058, fillcolor=black, dimen=outer](1.8,1.84)(0.2,0.2)
\psellipse[linecolor=black, linewidth=0.04, shadow=true,shadowsize=0.0058, fillcolor=black, dimen=outer](1.8,-1.76)(0.2,0.2)
\psdots[linecolor=black, dotsize=0.4](0.8,0.04)
\rput[bl](1.6,2.3){$2$}
\rput[bl](0.1,0.2){$2$}
\rput[bl](1.6,-2.5){$2$}
\end{pspicture}
}
\caption{Pictorial description of $\theta_2\theta_4^{-1}$}
\label{fig:14Inv}
\end{figure}

So, we may assume that $\sigma$ has a simple zero at the origin. Indeed, if $\sigma$ has three zeros as in Fig. \ref{fig:WSInd3}, we divide it by $\theta_2^{-1}\theta_1^{2}$ to obtain a theta function with a single zero at the origin. But this is not allowed as $V_1$ is 1-dimensional and whose generator has zero at $\frac{1}{2}(1+\tau)= \frac{1}{2}(1+\xi)$ by Proposition \ref{prop:zero-psi} below. 
Thus this case never occurs and the proof is complete.

 To prove the assertion above, we pass the problem back to linear solutions $\psi$ from the theta functions via (\ref{psi-nullL}) and use the complexified co-ordinates $x=x_1+i x_2$. 
We have
  \begin{proposition}\label{prop:zero-psi} Let $\psi$ satisfy the gauge-periodicity condition \eqref{gaugeperiod-psi}. Then it has zero at $\frac{1}{2}(1+\tau)$.  In particular, when $n=1$, $\psi$ does not vanish at $0$. Thus by uniqueness of theta functions, we conclude  there is no theta function with a single zero at the origin (with our imposed boundary condition). 
    \end{proposition}  
 \begin{proof} 
 Let us denote by $z = \frac{1}{2}(1+\tau)$. Due to the quasiperiodic boundary conditions 
\begin{align}
	&\psi(y+1)=e^{\frac{ikny_2}{2}}\psi(y) \\
	&\psi(y+\tau)=e^{\frac{ikn(\tau_1y_2-\tau_2y_1)}{2}}\psi(y)
\end{align}
Applying these relations at the point $z=(-(\tau_1+1)/2,-\tau_2/2)$, we find that
\begin{align}
 &\psi(z+1)=e^{-\frac{ikn\tau_2}{4}}\psi(z) \\
 &\psi(z+\tau)=e^{\frac{ikn\tau_2}{4}}\psi(z)
\end{align}
Now, utilizing the symmetry $\psi(-x)=\psi(x)$, we deduce that $\psi(z+1)=\psi(z+\tau)$. Thus
$$e^{\frac{ikn\tau_2}{4}}\psi(z)=e^{\frac{-ikn\tau_2}{4}}\psi(z),$$
or equivalently
$$e^{\frac{ikn\tau_2}{2}}\psi(z)=\psi(z).$$
Since $k\tau_2=2\pi$, this relation becomes $(e^{in\pi}-1)\psi(z)=0$, and implies when $n$ is odd that $\psi$ vanishes at $z$.  \end{proof}

\appendix


\section{On solutions of the linearized problem}\label{sec:no-sol}

\begin{lemma}\label{lem:no-sol}
There is no linear solution $\psi$, as in Section \ref{sec:operators}, such that
\begin{align}
	\psi(\bar{z}) = e^{ig_r(z)} \psi(z) \label{eqn:ReflSym}
\end{align}
for some real valued $g_r$.
\end{lemma}
\begin{proof}
Assume for the sake of contradiction that such $\psi$ exists. Then
\begin{align}
	\psi(z) = e^{\frac{n}{4}(z^2-|z|^2) }\theta(z)
\end{align}
for some holomorphic $\theta$. Equation (\ref{eqn:ReflSym}) becomes
\begin{align}
	\theta(\bar{z}) = e^{\frac{n}{4}(z^2-\bar{z}^2)} e^{ig_r(z)} \theta(z)
\end{align}
Taking $\partial_{z}$ on both sides, we see that
\begin{align}
	0 = e^{\frac{n}{4}(z^2-\bar{z}^2)} e^{ig_r(z)} (-\frac{n}{2}\bar{z} \theta+i\theta \partial_z g_r + \theta')
\end{align}
This shows that the term in the bracket vanishes identically. In particular,
\begin{align}
	(-\frac{n}{2}\bar{z} +i \partial_z g_r)\theta = -\theta'
\end{align}
Taking $\partial_{\bar{z}}$ again, we see that
\begin{align}
	(-\frac{n}{2} + i \partial_{\bar{z}}\partial_z g_r) \theta = 0
\end{align}
Since $\theta$ has at most finitely many zeros, we conclude
\begin{align}
	-\Delta g_r = 2n i
\end{align}
This is absurd since $g_r$ is real valued: a contradiction.
\end{proof}

\section{Theta Functions}\label{sec:theta-fns}
In this appendix we review basic properties of theta functions, which are likely to be known but which 
  we could not find in the literature. From now on, we fix a lattice shape $\tau$ and a lattice
\begin{align}
	\lat_\tau = \F{Z} + \tau \F{Z}
\end{align}
throughout this appendix (unless otherwise stated). 

\subsection{Basic Properties}
In this section, we prove some basic properties of the theta functions. Let $n$ be fixed. 
Define for $0 \leq m \leq n-1$,
\begin{align}
	\theta_{n,m}(z) = \sum_{l \in [m]_n} \gamma^{l^2} e^{2\pi i l z} \label{thetaBasis}
\end{align}
where $\gamma := e^{\pi i \tau/n}$ and $[m]_n = \{ a \in \F{Z} \ : \ a = m \bmod n \}$. 

\begin{theorem} \label{Basis}
The $\theta_{n,m}$'s form a basis for $V_n$ that satisfy
\begin{enumerate}
	\item $\theta_{n,m}(z+\frac{1}{n}) = e^{2\pi i m/n} \theta_{n,m}(z)$
	\item $\theta_{n,m}(-z) = \theta_{n-m}(z)$
	\item $\theta_{n,m}(z+\tau/n) = \gamma^{-1} e^{2\pi i z} \theta_{n,m+1}(z)$
\end{enumerate}
\end{theorem}

\begin{theorem} \label{unique}
Any $n$-theta function has exactly $n$ zeros modulo translation by lattice elements. Moreover, any two theta functions that share the same zeros (counting multiplicity) are linearly dependent. 
\end{theorem}

\begin{theorem} \label{PS}
$\theta_{1,0}$ has a simple zero at $\frac{1}{2}(1+\tau)$
\end{theorem}
\begin{proof}
See Proposition \ref{prop:zero-psi}. 
\end{proof}

\begin{theorem} \label{property}
Suppose that $\theta \in V_n$ and $\sigma \in V_m$, then $\theta\sigma \in V_{n+m}$.
\end{theorem}
\begin{proof}
Inspection.
\end{proof}

The proof of the theorems consists of the following lemmas:

\begin{proof}[Proof of Theorem \ref{Basis}]
Expanding in $e^{2\pi i k z}$ for $k \in \F{Z}$, the coefficients of any elemetn of $V_n$ satisfies the recurssion $c_{m+n} = c_m e^{i(2m+n)\pi \tau}$. This recursion implies that for $0 \leq m \leq n-1$, we have that
\begin{align}
	c_{m+ln} = c_m e^{i\pi\tau(l^2n+2lm)}
\end{align}
where $l$ is an integer. So the functions
\begin{align}
	\sum_{k \in \F{Z}} e^{i\pi\tau(k^2n+2km)} e^{2\pi i(nk +m)z} ,\ m =0,...,n-1
\end{align}
form a basis for the eigenspace. If we let $l = kn+m$, then we can rewrite the above as
\begin{align}
	\sum_{l \in [m]_n} e^{i\pi\tau \frac{l^2-m^2}{n}} e^{2\pi ilz} = e^{-i\pi m^2/n} \theta_m
\end{align}

Now we prove the three bullet points. We note that
\begin{align}
	\theta_m(z+\frac{1}{n}) =\sum_{l \in [m]_n} \gamma^{l^2} e^{2\pi i l z} e^{2\pi i l/n}
\end{align}
Since $l \in [m]_n$, we have that $l/n - m/n \in \F{Z}$. Hence
\begin{align}
	\theta_m(z+\frac{1}{n}) = e^{2\pi i m/n} \theta_m(z)
\end{align}

Now for the second item, we note
\begin{align}
	\theta_m(-z) &= \sum_{l \in [m]_n} \gamma^{l^2} e^{-2\pi i l z} \\
		&= \sum_{l \in [m]_n} \gamma^{(-l)^2} e^{2\pi i (-l) z} \\
		&= \sum_{l \in [n-m]_n} \gamma^{l^2} e^{2\pi i l z} \\
		&= \theta_{n-m \bmod n}(z)
\end{align}

Finally, recalling that $\gamma = e^{\pi i \tau/n}$, we note that
\begin{align}
	\theta_m(z+\tau/n) &= \sum_{k \in \F{Z}} \gamma^{(kn+m)^2} e^{2\pi i(kn+m)z+2\pi i(kn+m)\tau/n} \\
		&= \gamma^{-1} \sum_{k \in \F{Z}} \gamma^{k^2n^2+2knm+m^2+2kn+2m+1} e^{2\pi i(kn+m)z} \\
		&= \gamma^{-1} \sum_{k \in \F{Z}} \gamma^{(kn+m+1)^2} e^{2\pi i(kn+m)z} \\ 
		&= \gamma^{-1} e^{-2\pi i z} \theta_{m+1 \bmod n}(z)
\end{align}
\end{proof}

\begin{proof}[Proof of Theorem \ref{unique}]
First we prove that elements of $V_n$ has exactly $n$ zeros modulo translation by lattice elements. We compute the winding number of $\theta$. First, since $\theta$ is holomorphic, its zeros are discrete. Hence we may assume WLOG that all the zeros are in the interior of the fundamental domain. Let $\Omega$ denote the fundamental domain. Then the total number of zeros of $\theta$ is given by 
\begin{align}
	\frac{1}{2\pi i} \int_{\partial \Omega} \frac{\theta'}{\theta} dz
\end{align}
Since $\theta(z) = \theta(z+1)$, the integral along the $t\tau$ and $t\tau + 1$ for $t \in [0,1]$ is zero. Let $y(z) = e^{-i\pi \tau} e^{\alpha z}$ where $\alpha = -2\pi i$. Since $\theta(z+\tau) = y^n \theta(z)$ and $y' = \alpha y$, we see that $\theta'(z+\tau) = y^n(z) \theta'(z) + n \alpha y^n(z) \theta(z)$. Hence, only the horizonal segment of the line integral contribute:
\begin{align}
	\frac{1}{2\pi i} \int_{\partial \Omega} \frac{\theta'}{\theta} dz &= \frac{1}{2\pi i}\int_0^1 \frac{\theta'(t)}{\theta(t)} - \frac{\theta'(1+\tau - t)}{\theta(1+\tau-t)} dt \\
		&= \frac{1}{2\pi i} \int_0^1 \frac{\theta'(t)}{\theta(t)} - \frac{y^n(1-t) \theta'(1- t) + n \alpha y^n(1-t) \theta(1-t)}{y^n(1-t) \theta(1-t)} dt \\
		&= \frac{1}{2\pi i} \int_0^1 \frac{\theta'(t)}{\theta(t)} - \frac{\theta'(1- t) + n \alpha \theta(1-t)}{\theta(1-t)} dt \\
		&= \frac{1}{2\pi i} \int_0^1 \frac{\theta'(t)}{\theta(t)} - \frac{\theta'(1- t)}{\theta(1-t)} dt + n \\
		&= n
\end{align}
Next, we show that any two theta functions that share the same zeros (counting multiplicity) are linearly dependent. Let $\theta$ and $\varphi$ be the two nonzero zeta functions that shares the same zeros. Set $f(z) = \theta(a)/\varphi(z)$. We show that
\begin{enumerate}
	\item $f(z)$ can be extended analytically to all of $\F{C}$ and
	\item $f(z)$ is doubly periodic.
\end{enumerate}
Certainly $f$ is holomorphic away from zeros of $\varphi$. We only need to show that $f$ can be extended analytically to zeros of $\varphi$. But this is precisely the requirement that $\theta$ and $\varphi$ share the same zeros (counting multiplicity).

For the second item, we note that
\begin{align}
	f(z+1) =\theta(z+1)/\varphi(z+1) = \theta(z)/\varphi(z) = f(z)
\end{align}
and
\begin{align}
	f(z+\tau)  = \frac{\theta(z+\tau)}{\varphi(z+\tau)} = \frac{e^{-2\pi i nz} e^{-\pi i n \tau} \theta(z)}{e^{-2\pi i nz} e^{-\pi i n \tau} \varphi(z)} = \frac{\theta(z)}{\varphi(z)} = f(z)
\end{align}
This shows that $f$ is doubly periodic.

Now, Liouville's theorem shows that $f$ must be constant. It follows that $\theta$ and $\varphi$ are collinear.
\end{proof}

\subsection{Classification of singular $n$-theta functions}
\begin{theorem} \label{singular}
Let $X_n$ be the set of singular $n$-theta functions mod scaling. Then
\begin{align}
	X_n = \left\{ \theta_0^n\left( z + \frac{1}{n}(a+b\tau) \right) e^{2\pi i b z} : a,b \in \F{Z} \right\}
\end{align}
where $\theta_0$ is a basis for $V_1$. Moreover, $|X_n| = n^2$. The location of zeros of elements in $X_n$ form the set
\begin{align}
	\frac{1}{2}(1+\tau) + \frac{1}{n}(\F{Z} +\tau\F{Z})
\end{align}
\end{theorem}

As before, we establish the theorem through various lemmas. The idea of the proof is as follows: by Theorem \ref{unique}, we may identify elements of $X_n$ with the location of their zeros. We attempt to locate the zeros of singular $n$-theta function first and show that there are only $n^2$ possible locations in a fundamental cell. So $|X_n|=n^2$. Then we explicitly construct $n^2$ singular $n$-theta functions to complete the proof.

To locate the zeros of singular $n$-theta functions, we study the Wronskian of a particular set of nice basis element: $\Theta(z) := \det( \theta^{(i)}_j )$ for $i,j \in \{0,...,n-1\}$, where $\theta^{(i)}_j$ means the $i$-th derivative of $\theta_j$ (see equation \eqref{thetaBasis} for definition $\theta_j$). 

\begin{proposition} \label{Wronskian}
The function $\Theta$ is holomorphic and
\begin{enumerate}
	\item The locations of the zeros of $\Theta$ are exactly the locations where a singular $n$-theta function can have zero.
	\item $\Theta(-z) = (-1)^{n+1} \Theta(z)$,
	\item $\Theta(z+1/n) = (-1)^{n+1} \Theta(z)$,
	\item $\Theta(z+\tau/n) = (-1)^{n+1}\gamma^{n(n-1)} y^n \Theta(z)$ where $y = e^{-i\pi \tau}e^{\alpha z}$ and $\alpha = -2\pi i$.
\end{enumerate}
\end{proposition}
\begin{proof}
We recall that the $\theta_m$'s form a basis for $V_n$. If $\theta(z) = \sum_{m} a^m \theta_m(z)$ has $n$ zeros at $z_0$, then 
\begin{align}
	0 = \theta^{(i)}(z_0) = \sum_{m} a^m \theta^{(i)}_m(z)
\end{align}
for $i=0,...,n-1$. So the matrix $(\theta^{(i)}_j(z_0))$ has a nonzero vector $(a^0,...,a^{n-1})$ in its kernel. Hence $\Theta(z_0) = 0$. Conversely, if $\Theta(z_0)=0$, then we can find a nonzero vector $(a^0,....,a^{n-1})$ in the kernel of the matrix $(\theta^{(i)}_j(z_0))$. Then $\theta = a^m \theta_m$ has $n$-zeros at $z_0$. 

Recall from Theorem \ref{Basis} that $\theta_m(-z) = \theta_{n-m \bmod n}(z)$. It follows that $\theta^{(k)}_m(-z) = (-1)^k\theta^{(k)}_{n-m \bmod n}(z)$. If $n$ is even, then after $z \mapsto -z$, every even row in the matrix $(\theta^{(i)}_j)$ picks up a minus sign, and moreover, we need to interchange the $m$-th collumn with the $(n-m \bmod n)$-th collumn for $0 < m < n/2$. Together we pick up $n/2+n/2-1$ minus signs for $\Theta$. So $\Theta(-z) = -\Theta(z)$. If $n$ is odd, we pick up $(n-1)/2$ minus signs from the even rows and need to interchange $(n-1)/2$ columns. So $\Theta(-z) = \Theta(z)$.

Recall from Theorem \ref{Basis} that $\theta_m(z+1/n) = \zeta^m \theta_m(z)$ where $\zeta = e^{2\pi i /n}$. It follows after $z \mapsto z+1/n$, the $m$-th column of $(\theta^{(i)}_j)$ picks up a factor of $\zeta^{m-1}$. Hence $\Theta(z+1/n) = \zeta^{\sum_{k=0}^{n-1} k}\Theta(z) = (-1)^{n+1}\Theta(z)$.

Finally, we recall from Theorem \ref{Basis} and the definition $y =  e^{-i\pi \tau}e^{2\pi i z} = \gamma^{-n} e^{2\pi i z}$ that
\begin{align}
	\theta_m(z+\tau/n) =& \gamma^{-1} e^{2\pi i z} \theta_{m+1 \bmod n}(z) \\
		=& \gamma^{-1} \gamma^n y\theta_{m+1 \bmod n}(z) \\
		=& \gamma^{n-1} \theta_{m+1 \bmod n}(z) 
\end{align}
Repeated differentiation shows that
\begin{align}
	\theta_m^{(k)}(z+\tau) = \gamma^{n-1}  \sum_{i=0}^k {k \choose i} (y)^{(i)} \theta_{m+1}^{(k-i)}(z)
\end{align}
Hence
\begin{align}
	(\theta^{(i)}_j(z+\tau/n)) = \gamma^{n-1} E\left( \begin{array}{cccccc}
																		y \\
																		(y)' & y \\
																		(y)'' & 2(y)' & y \\
																		\vdots &&& \ddots \\
																		(y)^{n} && \cdots && y 
																	\end{array} \right) (\theta^{(i)}_j(z))
\end{align}
where $E$ is the matrix that corresponds to a permutation of collomns $(1,2,...,n) \mapsto (2,3,...,n,1)$. It follows that
\begin{align}
	\det &(\theta^{(i)}_j(z+\tau/n)) \notag\\ 
	&= (-1)^{n+1}\det \left[ \gamma^{n-1}  \left( \begin{array}{cccccc}
																		y \\
																		(y)' & y \\
																		(y)'' & 2(y)' & y \\
																		&&& \ddots \\
																		(y)^{n} &&&& y 
																	\end{array} \right) (\theta^{(i)}_j(z)) \right]
\end{align}
(where $(-1)^{n+1} = \det E$). Hence $\Theta(z+\tau) = (-1)^{n+1}\gamma^{n(n-1)} y^{n}\Theta(z)$.
\end{proof}

\begin{corollary}
$\Theta \in V_{n^2}$
\end{corollary}
\begin{proof}
The lemma above shows that
\begin{align}
	\Theta(z+1) =& \Theta(z+ \sum_{i=1}^n 1/n ) = (-1)^{(n+1)n} \Theta(z) = \Theta(z) 
\end{align}
We repeat the above proof with $\tau/n$ replaced by $\tau$. Note first that $\theta(z+\tau) = e^{-2\pi i nz -\pi i n \tau} \theta(z)$ for all $\theta \in V_n$. Set $Y = e^{-2\pi i nz -\pi i n \tau}$, then we see that
\begin{align}
	(\theta^{(i)}_j(z+\tau)) = \left( \begin{array}{cccccc}
																		Y \\
																		(Y)' & Y \\
																		(Y)'' & 2(Y)' & Y \\
																		\vdots &&& \ddots \\
																		(Y)^{n} && \cdots && Y 
																	\end{array} \right) (\theta^{(i)}_j(z))
\end{align}
Taking $\det$ of both sides, we see that $\Theta(z+\tau) = Y^n \Theta(z) = e^{-2\pi i n^2 z-\pi i n^2} \Theta(z)$, which is precisely the defining conditions of elements of $V_{n^2}$.
\end{proof}

\begin{corollary}
$|X_n| = n^2$.
\end{corollary}
\begin{proof}
The uniqueness theorem \ref{unique} shows us that $|X_n|$ is equal to the number of possible locations of zeros of singular $n$-theta functions. Proposition \ref{Wronskian} shows that that this is equal to the size of the zero set of $\Theta$ mod $L_\tau$. Since $\Theta \in V_{n^2}$. We conclude by Theorem \ref{unique}, again, that $|X_n| = n^2$.
\end{proof}

Now, we obtain explicit formuli for elements of $X_n$. To do this, we need the following lemma

\begin{lemma} \label{nsqr}
If $\theta \in V_n$, so is
\begin{align}
	\gamma(z) = \theta\left( z + \frac{1}{n}(a+b\tau) \right) e^{2\pi i b z}
\end{align}
for $a,b \in \F{Z}$.
\end{lemma}
\begin{proof}
We check that
\begin{align}
	\gamma(z+1) &= \theta\left( z + \frac{1}{n}(a+b\tau) +1  \right) e^{2\pi i b z+2\pi ib} \\
		&= \gamma(z)
\end{align}
since $b \in \F{Z}$. Similarly,
\begin{align}
	\gamma(z+\tau) &= \theta\left( z + \frac{1}{n}(a+b\tau) +\tau \right) e^{2\pi i b z+2\pi i b \tau} \\
		&= e^{-\pi i n \tau - 2\pi i n z - 2\pi i (a+b\tau)} \theta\left( z + \frac{1}{n}(a+b\tau) \right) e^{2\pi i b z +2\pi i b \tau} \\
		&= e^{-\pi i n \tau - 2\pi i n z} \gamma(z)
\end{align}
since $a,b \in \F{Z}$.
\end{proof}

Now, let $\theta_0$ be a basis for $V_1$. From theorem \ref{PS} and \ref{property}, we see that that $\theta_0^n \in X_n$, it follows by lemma \ref{nsqr} that
\begin{align}
	\theta_{a,b}(z):=\theta_0^n\left( z + \frac{1}{n}(a+b\tau) \right) e^{2\pi i b z}
\end{align}
are all in $X_n$ for $a,b \in \F{Z}$. But there are exactly $n^2 = |X_n|$ number of distinct such functions (mod scaling). So $X_n$ is contains exactly these elements. Moreover, by Proposition \ref{prop:zero-psi}, the zero of $\theta_0$ is at $\frac{1}{2}(1+\tau)$. So the zeros of $\theta_{a,b}$ are located at $\frac{1}{2}(1+\tau)-\frac{1}{n}(a+b\tau)$.

\section{Choice of $\chi_g$} \label{sec:H-equiv} 
The action of point groups is given by 
\begin{align}
	\psi(g x) = e^{i\chi_g} \psi(x). \label{gaugeEquiCondition}
\end{align}
for some $\chi_g$, which we determine below.

\begin{proposition}\label{prop:xig-xindep}
Let $g \in SH(\mathcal{L})$ and $\psi$ is a linear solution satisfying (\ref{gaugeEquiCondition}), then $\chi_g$ are constant.
\end{proposition}
\begin{proof}
We identify $SH(\mathcal{L})$ as a subset of $\C$ so that $gx$ is the multiplication of the two complex numbers $g$ and $x$. Assume that $\chi_g$ satisfies (\ref{gaugeEquiCondition}). Since $\psi$ is a linear solution, by (\ref{psi-nullL}), we can find a holomorphic theta function $\theta$ such that $\theta(x) = h(x)\psi(x)$ for some smooth, nonvanishing, $h$ with the property $(\bar \partial h)(x) = \frac{b}{2}	xh(x)$. Then (\ref{gaugeEquiCondition}) is equivalent to the fact that
\begin{align}
	H_g(x) := h(g x) e^{i\chi_g} h(x)^{-1}
\end{align}
is holomorphic. Taking $\dbar$, this requirement is equivalent to
\begin{align}
	0=& \dbar(h(g x) e^{i\chi_g} h( x )^{-1}) \\
	=& (i\dbar \chi_g + \frac{b}{2}\bar{g}gx - \frac{b}{2}x)h(g x) e^{i\chi_g} h(x)^{-1} .
\end{align}
Since $|g|=1$ and $h(g x) e^{i\chi_g} h(x)^{-1}$ is invertible, we see that
\begin{align}
	\dbar \chi_g = 0
\end{align}
Since $\chi_g$ are real valued, it is a constant.
\end{proof}

As a result of the the proposition, it suffices for us to look for gauge invariant $(\psi, A)$ under actions of $H(\cL)$ whose gauge factor $h_g(x)=e^{i \chi_g}$ is a constant. Hence we consider spaces of the form
\begin{align}
	\{ \psi(R_\xi^{-1} i x) = \eta \psi(x),\  R_\xi A(R_\xi^{-1} x) = \eta' A(x) \}
\end{align}
where $\eta, \eta' \in \F{C}$. One realizes that such space corresponds to irreducible representations of $H(\cL)$.

\section{Table of $C_6$-equivariant Theta Functions}  \label{app:ThetaTable}
\begin{align} \label{theta-table}
\begin{array}{ | c | c | c | } 
	\hline 
	\text{Vortex Number} & 
	 \text{Value of } r & \text{Theta functions that span $V_{n,6,r}$} \\ 
	\hline 
	n = 2 	& 0 			& \theta_0 \\
					& 2 	& \theta_2 \\ 
	\hline 
	n = 4		& 0				& \theta_0^2 \\
					& 1			& \theta_1 \\
					& 2		& \theta_0 \theta_1 \\
					&	4		& \theta_2^2 \\
	\hline
	n = 6 	& 0				& \theta_0^3, \ \theta_2^3, \theta_1^2 \theta_2^{-1} \\
					& 1			& \theta_0 \theta_1 \\ 
					& 2		& \theta_0^2 \theta_2 \\
					& 3		& \theta_1 \theta_2 \\
					& 4		& \theta_0 \theta_2^2 \\
	\hline
	n = 8 	& 0				& \theta_0^4, \ \theta_0 \theta_2^3 \\
					& 1			& \theta_0^2 \theta_1 \\
					& 2		& \theta_2^4, \ \theta_1^2, \ \theta_0^3 \theta_2 \\
					& 3		& \theta_0 \theta_1 \theta_2 \\
					& 4		& \theta_0^2 \theta_2^2 \\
					& 5		& \theta_1 \theta_2^2 \\
	\hline
	n = 10 	& 0				& \theta_0^5,\ \theta_0^2 \theta_2^3 \\
					& 1			& \theta_0^3 \theta_1, \ \theta_1 \theta_2^3 \\
					& 2		& \theta_0^4 \theta_2,\ \theta_0 \theta_2^4, \ \theta_0 \theta_1^2 \\
					& 3		& \theta_0^2 \theta_1 \theta_2 \\
					& 4		& \theta_0^3 \theta_2^2 , \ \theta_2^5, \ \theta_1^2 \theta_2 \\
					& 5		& \theta_0 \theta_1 \theta_2^2 \\
	\hline
\end{array}
\end{align}



\end{document}